\documentclass[11pt]{article}
\usepackage[utf8]{inputenc}
\usepackage[skins,xparse,breakable]{tcolorbox}
\usepackage{xcolor}
\definecolor{linkblue}{HTML}{001487}
\usepackage{graphicx} 
\usepackage{stmaryrd}
\usepackage{amsmath, amssymb}
\usepackage{mdframed}
\usepackage{mathtools}
\usepackage{braket}

\usepackage{graphicx}

\usepackage[margin=1in]{geometry}
\usepackage{amsmath,amssymb,amsthm,mathtools}
\usepackage{thm-restate}
\usepackage{booktabs}
\usepackage{algorithm}
\usepackage{algpseudocode}
\algrenewcommand{\algorithmicrequire}{\textbf{Input:}}
\algrenewcommand{\algorithmicensure}{\textbf{Output:}}
\algnewcommand{\lIf}[2]{\State\algorithmicif\ #1\ \algorithmicthen\ #2}
\usepackage[colorlinks=true,linkcolor=blue,citecolor=blue,urlcolor=blue]{hyperref}

\theoremstyle{plain}
\newtheorem{theorem}{Theorem}[section]

\newtheorem{lemma}[theorem]{Lemma}
\newtheorem{proposition}[theorem]{Proposition}

\newtheorem{fact}[theorem]{Fact}
\theoremstyle{definition}
\newtheorem{definition}{Definition}

\newcommand{\len}{\mathrm{len}}
\newcommand{\poly}{\mathrm{poly}}
\renewcommand{\d}{\mathrm{d}}

\newcommand{\SU}{\mathrm{SU}}
\newcommand{\R}{\mathbb{R}}

\newcommand{\val}{\operatorname{val}}

\newcommand{\Pthirtytwo}{\mathcal{P}_{32}}

\usepackage{authblk}

\begin{document}

\title{Breaking the cubic barrier for the inverse-free \\ Solovay-Kitaev algorithm}

\date{}

\clearpage

\author[1]{Adam Bouland\thanks{ \texttt{abouland@stanford.edu}}}
\author[1]{Victor Shyaka\thanks{ \texttt{victorks@stanford.edu}}}
\author[1]{Chenyi Zhang\thanks{ \texttt{chenyiz@stanford.edu}}}

\affil[1]{Department of Computer Science, Stanford University}

\maketitle
\thispagestyle{empty}

\begin{abstract}
The Solovay-Kitaev algorithm states that it is possible to efficiently compile unitaries to accuracy $\epsilon$ in $\text{polylog}(1/\epsilon)$ time using any universal gate set. While Solovay and Kitaev's initial work addressed inverse-closed gate sets, recently Bouland and Giurgic\v{a}-Tiron showed it is possible to efficiently compile  arbitrary gate sets without inverses. However, the exponent of their compilation algorithm is high. For example the inverse-free exponent for a qubit is over $8.62$, while information theoretically it is known that an inverse-free exponent of $3$ is possible in any dimension by a result of Oszmaniec, Sawicki, and Horodecki. This leads to a natural question: can one improve the exponent in Bouland and Giurgic\v{a}-Tiron's algorithm? And are there any cases where the inverse-free exponent is below 3, thus surpassing existing information-theoretic arguments?

To answer this question, 
we first show that a special case of the inverse-free compiling problem -- namely Pauli gates supplemented by an irrational rotation of a qubit - admits a highly efficient compilation algorithm with an exponent at most $\approx 2.138$. 
Our algorithm fuses ideas from prior algorithms of Kuperberg and Sardharwalla, Cubitt, Harrow and Linden to construct an efficient ``Pauli golf'' compilation routine. We then extend these ideas to the general inverse-free case for a qubit, where we give a compiling algorithm with  exponent at most $2.988$. This result uses a novel form of approximate ``aiming'' of unitary corrections to improve the compilation efficiency. Our results give two examples of quantum compiling algorithms surpassing information-theoretic arguments, and demonstrate that there are many efficiency gains possible for inverse-free compilation.

\end{abstract}

\newpage \section{Introduction}
\setcounter{page}{1}

\vspace{-0.5em}

The Solovay-Kitaev algorithm is a fundamental result in quantum computing theory.
It states that it is possible to efficiently compile any universal quantum gate set into any other, with only polylogarithmic overhead.
More formally, for any unitary $U$ of fixed dimension, any universal gate set $G$, and any $\epsilon>0$, there is a sequence of gates from $G$ which $\epsilon$-approximates $U$ of length $\text{polylog}(1/\epsilon)$, and moreover there is an efficient algorithm to find this sequence.
This result is fundamental to quantum computing, as it implies that the runtime of quantum algorithms is independent of the choice of universal gate set up to polylogarithmic factors \cite{nielsen2000quantum,dawson2006solovay}.

The initial proofs of the Solovay-Kitaev algorithm required that gate set be inverse-closed, i.e. if a gate $g\in G$ then also $g^\dagger \in G$.
In this setting the best exponent for the polylog was $3$ for over 20 years by the algorithm of Kitaev, Shen and Vyalyi \cite{kitaev2002classical}. Subsequently in 2023 Kuperberg improved the general inverse-full exponent to $\approx1.440$ \cite{kuperberg2023breaking}, which Yehuda  then improved to  $\approx1.138$ for the special case of a qubit \cite{yehuda2026improved}.
There have also been numerous improvements to these algorithms for particular gate sets like Clifford + T which are important in fault-tolerance protocols, where optimal $O(\log(1/\epsilon))$ compilation algorithms are known e.g. \cite{ross2014optimal,kliuchnikov2014asymptotically,paetznick2013repeat,bocharov2015efficient}.
Interestingly it is known that for any algebraic gate set, that $O(\log(1/\epsilon))$ length sequences exist \cite{bourgain2010spectral,harrow2002efficient} - the challenge is ``merely'' to find them!

In the case the gate set is not closed under inverses, much less is known.
The first inverse-free Solovay-Kitaev (IFSK) algorithm was shown by Bouland and Giurgic\v{a}-Tiron in 2021 \cite{bouland2021efficient}, building on prior partial progress of Sardharwalla, Cubitt, Harrow and Linden \cite{sardharwalla2016universal} and Bouland and Ozols \cite{bouland2017trading}.
Bouland and Giurgic\v{a}-Tiron's IFSK algorithm is far from optimal however.
In particular, Oszmaniec, Sawicki and Horodecki showed that for any (non inverse closed) gate set, that $O(\log^3(1/\epsilon))$ sequences exist to approximate any unitary \cite{oszmaniec2021epsilon}, building on spectral gap work of Varju \cite{varju2013random}.
However Bouland and Giurgic\v{a}-Tiron's algorithm's exponent is much higher, being higher than 8 for a qubit, and growing with the dimension. 
The only lower bound on the length the sequence is $\Omega(\log(1/\epsilon))$ by  counting arguments.
This leads to a natural question: can one improve Bouland and Giurgic\v{a}-Tiron's algorithm? In particular, are there any cases where one can bring the exponent below three, thus surpassing current information-theoretic arguments?

\subsection{Our results}

We answer this problem in the affirmative. In particular we describe two different IFSK algorithms with exponents below three.
First, we show that if your gate set contains the Pauli gates, that there exists an efficient IFSK algorithm with exponent $\approx2.138$.

\begin{theorem}[informal]\label{thm:free_pauli_compilation_informal} Let $G\subset\mathrm{SU}(2)$ be a finite universal gate set with $X,Y,Z\in G$. There is an efficient algorithm to $\epsilon$-approximate any unitary $U$ using $O(\log^{2.138}(1/\epsilon))$ many gates from $G$.
\end{theorem}

This is precisely the setting of Sardharwalla, Cubitt, Harrow and Linden's algorithm which made the first partial progress towards IFSK \cite{sardharwalla2016universal}.
This first algorithm interleaves ideas from several prior IFSK and SK algorithms to yield the improved exponent \cite{sardharwalla2016universal,kuperberg2023breaking,yehuda2026improved}.

Second, we show the general IFSK exponent for a qubit is at most $\approx2.988$, thus surpassing information-theoretic arguments for a generic inverse-free gate set:
\begin{theorem}[informal]\label{thm:general_inverse_free_informal} Let $G\subset\mathrm{SU}(2)$ be a finite universal gate set. There is an efficient algorithm to $\epsilon$-approximate any unitary $U$ using $O(\log^{2.988}(1/\epsilon))$ many gates from $G$.
\end{theorem}

These results imply that the information theoretic IFSK result of Oszmaniec, Sawicki, and Horodecki \cite{oszmaniec2021epsilon} is not tight for a qubit. 
We conjecture that further improvements to their arguments might be possible (beyond the exponents achieved by our algorithms) for the case of a qubit, and leave this as an open problem\footnote{We note our results do not directly imply any improvements to the spectral gap results, as a spectral gap is sufficient for but not necessary for an information-theoretic IFSK result.}.

\subsection{Proof sketch}

\paragraph{A highly efficient IFSK algorithm with Pauli gates:}
Our first result is to study the special case of the IFSK problem on a qubit, where one assumes the Pauli gates are included in the gate.
The prior IFSK algorithm of Sardharwalla, Cubitt, Harrow and Linden in this setting has an exponent of at least 3 \cite{sardharwalla2016universal}, as it was building on prior algorithms with higher exponents \cite{kitaev2002classical,dawson2006solovay}.

A natural idea for bringing the exponent below 3 is to use ideas from Kuperberg's algorithm for the inverse-ful case \cite{kuperberg2023breaking}. 
Kuperberg's algorithm works by a form of ``golf'' which proceeds in two steps.
The first step is ``stick generation'' where one generates unitary rotations (=``sticks" or ``golf strokes'') of specified lengths.
Kuperberg shows this can be done highly efficiently using self-cancelling words with inverses, with exponent around 1.44.
The second step is ``precision golf'' where one uses conjugation by unitaries to ``aim'' the golf strokes in particular directions to steer towards a target unitary.
Kuperberg shows this golf subroutine is doubly exponentially convergent, making its contribution to the exponent negligible.
Unfortunately both of these steps use gate inversion in a key way.
For example, aiming a golf stroke $g$ by conjugation (e.g. $g\rightarrow UgU^\dagger$) only preserves the length of $g$ if the inverse is performed \emph{exactly} - and hence there is not an obvious way to generalize this to the inverse-free case.
Likewise Kuperberg's stick generation algorithm makes inherent use of inverses to create highly efficient self-cancelling words.
Therefore new ideas are needed to use these subroutines in the inverse-free case.

The starting point for our first algorithm is that some form of ``aiming'' is still possible if your gate set contains the Paulis. Namely, by conjugating by Paulis, one can aim a given golf stroke roughly in the cardinal directions of the unitary group $SU(2)$ (see Figure \ref{fig:Golf}). We show that this weak form of aiming is enough to converge to a given point in the space - effectively because each stroke takes us a constant fraction of the remaining distance to the desired point. This reduces the problem to ``stick generation", i.e. generating smaller and smaller nontrivial unitary rotations.
However, we note this ``Pauli golf'' subroutine comes at a cost of 1 in the exponent of the algorithm. This is because it merely converges singly exponentially quickly to a desired point, whereas the precision golf subroutine of Kuperberg is doubly exponentially convergent. 

The next ingredient of our algorithm is to note that \cite{sardharwalla2016universal}'s work contains an ingredient which can be used for stick generation.
Their algorithm uses the fact that if $U$ is $\epsilon$-close close to the identity, then $f(U)=XUXYUYZUZU$ is $O(\epsilon^2)$-close to the identity.
This allows one to generate smaller and smaller ``sticks'' we can use in our Pauli golf subroutine.
Unfortunately this routine is not efficient enough to break the cubic barrier.
The cost of produce a stick of length $\epsilon$ via iterating $f(U)$ is $\log^2(1/\epsilon)$, where the exponent of two comes from the ratio of logarithms of the number of $U$'s invoked (4) to the cancellation order (2).
Adding one from the singly-exponentially convergent Pauli golf subroutine would produce an exponent of 3, so would not improve on info-theoretic arguments.

Our key insight is that this stick generation subroutine can be made more efficient. In particular, we show that the errors in the resulting ``stick'' produced by $f(U)$ tend to be larger in certain directions than others. In particular if $U$ is a rotation which is $\epsilon$-close to the identity, then 
\[f(U)= (1-O(\epsilon^4)) I + O(\epsilon^3) X + O(\epsilon^2) Y + O(\epsilon^4) Z\]
In other words, the resulting ``stick'' tends to point in the $Y$ direction, with suppressed errors in other directions.
This suggested an improved way of iterating the $f(U)$ construction to decrease the length of the resulting stick. In particular, by conjugating $f(U)$ by a Pauli, we can flip the signs of certain error directions of $f(U)$. This allows us to cancel out the larger directions in the resulting stick. 
For example one can easily show that
\[
Zf(U)Zf(U) = (1-O(\epsilon^8)) I + O(\epsilon^6) X + O(\epsilon^7) Y + O(\epsilon^4) Z,
\]
i.e. two iterations can yield a cancellation order of 4 using 8 copies of $U$.
The question then becomes how to optimally iterate the construction to get the best possible exponent.
Fascinatingly we show there is way to recurse $f(U)$ interspersed with Paulis to obtain exponent $\log(2)/\log(\rho) \approx 1.138$, 
where $\rho$ is the rate of exponential growth of the Tribonacci sequence, i.e. the sequence where the next item generated is the sum of the last three items. 
In our setting this occurs because we can pick two out of the three error directions to improve at this step, which we show naturally gives rise to Tribonacci numbers in the recursion.
We note this Tribonacci exponent previously appeared in work of Yehuda \cite{yehuda2026improved} for the inverse-ful case on a qubit, and our Tribonacci exponent appears for related reasons in the inverse-free Pauli case.
Putting this all together, our net exponent of our compilation algorithm is $\log(2)/\log(\rho) +1 \approx 2.138$ which proves the claimed bound.

\paragraph{A subcubic IFSK algorithm:}

Our second result brings the IFSK exponent for a qubit below 3 in the general case.
For this algorithm we cannot use a form of ``golf'' as we do not have access to any gates with inverses to ``aim'' using conjugation.
Moreover, using the self-correcting sequences of Bouland and Giurgic\v{a}-Tiron in a more efficient framework like \cite{kitaev2002classical} yields exponents above 3.
Therefore we need a completely different framework to break the cubic barrier.

Instead, our starting point is to consider a weaker form of ``aiming'' in our algorithm to improve the exponent.
Suppose that the difference between our current unitary (composed of gates of $G$) and our target is in some direction $e$ away of length $\epsilon$. Our first step is to produce a unitary $S$ whose distance from the identity is proportional to $\epsilon$ (up to constants) and whose direction \emph{approximately} aligns with a given unit vector $e$. Formally, writing $S=\exp[-i(s_1X+s_2Y+s_3Z)]$ with $s\in\mathbb R^3$, we require
\[
\epsilon\le\|s\|\le100\epsilon,
\qquad
\left\|\frac{s}{\|s\|}-e\right\|\le\frac1{20}.
\]
We call such an $S$ an $(e,\epsilon)$-directed unitary. 
We show that such a directed unitary can be constructed from a fixed product containing 32 copies each of two unitaries, which need only be approximated to much coarser precision, namely  $\epsilon^{2/15}/100$ and $\epsilon^{1/5}/100$, respectively. Although both errors are much larger than $\epsilon$, cancellations in the product give the stated bounds on $S$. 
The sequence is essentially a higher order variant of one of the self-correcting sequences employed in Bouland and Giurgic\v{a}-Tiron's algorithm, with modifications to better allow the resulting unitary to be ``aimed.''
This directed unitary, however, is not itself enough to make significant progress towards our target -- we cannot simply append $S$ to our gate sequence, as we do not have precise enough control on its length to ensure we make progress.

Instead, we show that we can use a directed unitary $S$ to construct the required correction to our gate sequence, substantially reducing our error. In particular, we define a sequence $\mathcal P_{24}(S,C_1,C_2)$ consisting of copies of $S$\footnote{Technically we use directed unitaries $S$ pointing in different directions, see main proof for details.}, $C_1$, and $C_2$, where $C_1,C_2$ are chosen in the neighborhood of two unitaries $Q_1,Q_2$ that are chosen in advance and are referred to as \textit{anchors}. ($Q_1,Q_2$ can be thought of as $X$ and $Z$, as it turns out anticommuting operators are best for this construction). We show that as $C_1,C_2$ vary near the anchors, we show the sequence attains every unitary within distance $c\epsilon^2$ of $\mathcal P_{24}(S,Q_1,Q_2)$, for an absolute constant $c>0$. In other words, this sequence is \emph{onto} the space of unitary corrections, with very small remaining error. In particular we show the remaining error in our correction, accounting for all errors in the subcomponents, is $O(
\epsilon^3
)$. 
In other words, if we can create directed unitaries, we can then use them to create \emph{arbitrary} unitary corrections to high accuracy within some region.

We then show that we can piece these two subroutines together to produce an overall algorithm with exponent below 3.
The key reason this is possible is that our recursive calls to our length-32 sequence for directed unitaries allow very large error, so can be drawn from much coarser approximations.
More formally, to bound the total sequence length, let $L(\epsilon)$ denote the worst-case output length at error $\epsilon$. The sequence $\mathcal P_{24}$ contains eight copies each of $S,C_1,C_2$, while constructing $S$ uses 32 copies of each of the two much coarser approximations. We show that using the error guarantees in Section~\ref{sec:inverse_free}, we obtain the following recursion
\[
\begin{aligned}
L(\epsilon^3)
\le L(c\epsilon^2)+16L(c\epsilon)
+256L(\epsilon^{2/15}/100)
+256L(\epsilon^{1/5}/100).
\end{aligned}
\]
The recursion stops at a fixed positive error threshold, where a precomputed finite list supplies a net of constant precision. With this initial bound, we show the recurrence yields $L(\epsilon)=O(\log^{2.988}(1/\epsilon))$. 
We note this algorithm is also computationally efficient.
There are $O(\log\log(1/\epsilon))$ recursion levels, with at most five recursive calls per step. Together with the $\poly\log(1/\epsilon)$ work per call, this gives a total classical running time of $\poly\log(1/\epsilon)$ as desired.

\section*{AI Statement}
The compilation algorithm with Pauli gates was human-generated with minor calculational assistance from GPT 5.6 Sol and Claude Opus.
The general inverse-free compilation algorithm for a qubit was developed using GPT 5.6 Sol and GPT 6 Astra, with many rounds of interaction and simplification required to bring the exponent below 3.

\section*{Acknowledgments}
This material is based on work supported in part by the National Science Foundation under award No. 2016245 and 2440805, by the Air Force Office of Scientific Research under grant agreement FA9550-21-1-0392 and FA9550-24-1-0089, by the DOE Office of Science under grant agreements DE-SC0020377 and DE-SC0025934, and by the Shoucheng Zhang graduate fellowship.

\section{Preliminaries}\label{sec:preliminaries}

\paragraph{General notation.} Denote $\mathcal S\coloneqq \{v\in\R^3\mid \|v\|=1\}$ to be the unit sphere in $\R^3$. For any $x\in\R^3$ and $r>0$, write $B_x(r)\coloneqq \{y\in\R^3\mid \|y-x\|\leq r\}$ for the closed ball of radius $r$ centered at $x$.

We define a \emph{word} $w=(g_1,\ldots,g_\ell)$ to be a finite classical string of gate symbols. For a gate set $G$, we call $w$ a \emph{$G$-word} if each symbol represents a gate in $G$. We write $\len(w)\coloneqq\ell$ for its length and $\val(w)$ for the product of the corresponding unitary matrices in the listed order. We denote words by lowercase letters. When there is no ambiguity, we use the corresponding uppercase letters for the unitaries they represent; for example, $W=\operatorname{val}(w)$.

\paragraph{Constant precision nets.} The following standard result provides a finite list of gate sequences for approximating any unitary in $\SU(2)$ to constant precision; see \cite[Section~3]{dawson2006solovay} and \cite[Section~2.2]{bouland2021efficient}.

\begin{lemma}\label{lem:compilation_base_case}
Let $G\subset\mathrm{SU}(2)$ be a finite universal gate set. For any fixed $\delta=\Omega(1)$, one can compute a finite list of $G$-words whose values form a $\delta$-net of $\SU(2)$ in $O(1)$ time. Equivalently, there is a procedure $\delta\text{-}\textsc{ConstantNet}$ that, given $U\in\SU(2)$, returns a word $w$ from this list in $O(1)$ time with $\len(w)=O(1)$ and $\|W-U\|\le\delta$.
\end{lemma}

\paragraph{Quaternion representation of $\SU(2)$. }
$\mathrm{SU}(2)$ is isomorphic to the unit-quaternion group
\[
\mathrm{Sp}(1)
:=
\left\{
(\tau,v) :
\tau\in\mathbb R,\ 
v\in\mathbb R^3,\ 
\tau^2+\lVert v\rVert^2=1
\right\},
\]
where the isomorphism we use in this paper is
\[
\Phi:\mathrm{Sp}(1)\longrightarrow\mathrm{SU}(2),
\qquad
\Phi(\tau,v)
=
\tau I-i(v_1X+v_2Y+v_3Z),
\]
with quaternion multiplication given by
\[
(\tau,v)(\tau',v')
=
\bigl(\tau\tau'-v\cdot v',\,
\tau v'+\tau'v+v\times v'\bigr).
\]
For a quaternion $(\tau,v)$, write $\operatorname{Re}(\tau,v)=\tau$ and $\operatorname{Im}(\tau,v)=v$. Moreover, for
$(\tau,v),(\tau',v')\in\mathrm{Sp}(1)$, we have $(\tau,v)^{-1}=(\tau,-v)$, and
\[
\left\lVert\Phi(\tau,v)-\Phi(\tau',v')\right\rVert_{\mathrm{op}}
=
\lVert(\tau,v)-(\tau',v')\rVert.
\]
For $u\in\mathbb R^3$, we write
\[
\exp(0,u)\coloneqq 
\begin{cases}
\displaystyle\left(\cos\|u\|,\frac{\sin\|u\|}{\|u\|}u\right),&u\neq0,\\[6pt]
(1,0),&u=0,
\end{cases}
\]
whose inverse operation is the principal logarithm defined as
\begin{equation}\label{eq:quaternion_principal_log}
\log(\tau,v)\coloneqq
\begin{cases}
\displaystyle\frac{\arccos\tau}{\|v\|}\,v,&v\neq0,\\[6pt]
0,&v=0.
\end{cases}
\end{equation}
for any $(\tau,v)\in\mathrm{Sp}(1)\setminus\{(-1,0)\}$.
For $U\in\mathrm{SU}(2)\setminus\{-I\}$, we write $\log U\coloneqq\log\bigl(\Phi^{-1}(U)\bigr)\in\mathbb R^3$.

\begin{lemma}
\label{fact:quaternion-rotation}
For every $Q=(\tau,v)\in\mathrm{Sp}(1)$, there is a unique rotation $R_Q$ in
$\mathbb R^3$ such that
\[
Q(0,u)Q^{-1}=(0,R_Qu),\qquad u\in\mathbb R^3.
\]
Its rotation angle $\theta\in[0,\pi]$ satisfies $\cos\theta=2\tau^2-1$. If $v\ne0$, its rotation axis is $v$. Moreover, if $Q=\exp(0,h)$ for some $h\in\mathbb R^3$ and $\|h\|\le1$, we also have
\[
R_Qx=x+2h\times x+O(\|h\|^2\|x\|),\qquad x\in\mathbb R^3.
\]
\end{lemma}

\begin{proof}
Note that
\[
\begin{aligned}
Q(0,u)Q^{-1}
&=(-v\cdot u,\;\tau u+v\times u)(\tau,-v)\\
&=\Bigl(0,\;(\tau^2-\|v\|^2)u
       +2(v\cdot u)v+2\tau(v\times u)\Bigr),
\end{aligned}
\]
where the second component is linear in $u$ and can be uniquely described by a linear operator $R_Q$. In the rest of the proof, we show that $R_Q$ is a rotation in $\R^3$.

For any $u,u'\in\mathbb R^3$, direct multiplication gives
\[
(0,R_Qu)(0,R_Qu')
=\bigl(-(R_Qu)\cdot(R_Qu'),\;(R_Qu)\times(R_Qu')\bigr).
\]
On the other hand, associativity gives
\[
\begin{aligned}
(0,R_Qu)(0,R_Qu')
&=Q(0,u)(Q^{-1}Q)(0,u')Q^{-1}\\
&=Q(-u\cdot u',\,u\times u')Q^{-1}\\
&=\bigl(-u\cdot u',\,R_Q(u\times u')\bigr).
\end{aligned}
\]
Comparing the scalar and vector components respectively, we have
\[
(R_Qu)\cdot(R_Qu')=u\cdot u',
\qquad
(R_Qu)\times(R_Qu')=R_Q(u\times u').
\]
Therefore, $R_Q$ is a rotation since it preserves inner products and orientation.

For $Q=\exp(0,h)$, the expansion $\exp(0,\pm h)=(1,0)\pm(0,h)+O(\|h\|^2)$ gives
\[
\begin{aligned}
(0,R_Qx)
&=\exp(0,h)(0,x)\exp(0,-h)\\
&=(0,x)+(0,h)(0,x)-(0,x)(0,h)+O(\|h\|^2\|x\|)\\
&=(0,x+2h\times x)+O(\|h\|^2\|x\|).
\end{aligned}
\]
Taking imaginary parts gives the claimed expansion.
\end{proof}

\paragraph{Fixed points.} Let $D\subseteq\R^3$ be a nonempty set, and let $F:D\to D$ be a map. A point $h\in D$ is a \emph{fixed point} of $F$ if $F(h)=h$. We write $\|\cdot\|$ for the Euclidean norm on $\R^3$.

\begin{lemma}[Banach's fixed-point theorem, see e.g., {\cite[Theorem~9.23]{rudin1976principles}}]\label{lem:banach_fixed_point}
Let $D\subseteq\R^3$ be nonempty and closed, and let $F:D\to D$ be a map. Suppose that there exists $q\in[0,1)$ such that
\[
\|F(x)-F(y)\|\le q\|x-y\|\qquad\text{for all }x,y\in D.
\]
Then there is a unique $h\in D$ satisfying $F(h)=h$.
\end{lemma}

\section{A highly efficient IFSK algorithm with Pauli gates}\label{sec:free_paulis}
In this section, we construct an inverse-free compiler for gate sets containing the Pauli gates. 

\begin{theorem}[Formal version of Theorem~\ref{thm:free_pauli_compilation_informal}]
Let $G\subset\mathrm{SU}(2)$ be a finite universal gate set with $X,Y,Z\in G$. Given $U\in\mathrm{SU}(2)$ and $0<\epsilon<1$, Algorithm~\ref{alg:compile} outputs a $ G$-word $w$ with $\len(w)=O(\log^{2.138}(1/\epsilon))$ and $\|W-U\|\le\epsilon$.
\label{thm:free_pauli_compilation_3}
\end{theorem}
Our algorithm improves a coarse approximation by repeatedly appending a correction that reduces the error by a fixed factor. Section~\ref{sec:pauli-golf} develops a subroutine that we call \textit{Pauli golf}, which combines three small unitary rotations (``sticks''), choosing their sizes and orienting them by Pauli conjugation so that their product gives the required correction. Section~\ref{sec:stick-gen} demonstrates how to construct these sticks of decreasing size needed for this process. Then, Section~\ref{sec:compilation-alg} combines these ingredients into Algorithms~\ref{alg:compile} and~\ref{alg:axis}, and bounds their output lengths and runtime.

\subsection{Pauli Golf}\label{sec:pauli-golf}
Here we show how to use Pauli gates to ``aim'' small corrections towards a target unitary. These corrections use \emph{Pauli sticks}, which approximate $e^{-ihP}$ for $P\in\{X,Y,Z\}$ and small $|h|$ with error at most a small fraction of $|h|$. Lemma~\ref{lem:pauli_golf} shows how to choose the sizes and Pauli conjugations of three sticks so that appending the product of their conjugates to a sufficiently close approximation reduces the error by at least a factor of two. We call this process \emph{Pauli golf} (Figure~\ref{fig:Golf}).

\begin{definition}\label{def:pauli_stick}
Let $\sigma\in\{-1,1\}$, $0<\epsilon<1$, $P\in\{X,Y,Z\}$, and $0<\eta<1/4$. We call a unitary $S\in\SU(2)$ a \emph{$(\sigma,\epsilon,P,\eta)$-stick} if $\|S-e^{-i\sigma\epsilon P}\|\le\eta\epsilon$.
\end{definition}

\begin{lemma}\label{lem:pauli_golf}
There exist constants $\eta_0,\epsilon_0\in(0,1/4)$ with the following property. For every $E\in\SU(2)$ satisfying $0<\|E-I\|\le\epsilon_0$, there exists $\|E-I\|^2\le\epsilon_P\le2\|E-I\|$ and $Q_P\in\{I,X,Y,Z\}$ for $P\in\{X,Y,Z\}$, such that any $(1,\epsilon_P,P,\eta_0)$-sticks $S_P$ (Definition~\ref{def:pauli_stick}) give
\[
\left\|
(Q_XS_XQ_X)(Q_YS_YQ_Y)(Q_ZS_ZQ_Z)-E
\right\|
\le\frac12\|E-I\|.
\]
Moreover, the values of $\epsilon_P$ and $Q_P$ for $P\in\{X,Y,Z\}$ can be computed in $O(1)$ time.
\end{lemma}

\begin{proof}
Write $r=\|E-I\|$ and $h=(h_X,h_Y,h_Z)=\log E$. For sufficiently small $\epsilon_0$,
\[
E=e^{-i(h_XX+h_YY+h_ZZ)},
\qquad
\|h\|=2\arcsin(r/2)\le2r.
\]
To keep all three sizes positive, set $\epsilon_P=\max\{|h_P|,r^2\}$ for $P\in\{X,Y,Z\}$. Then $r^2\le\epsilon_P\le2r$ and $|\epsilon_P-|h_P||\le r^2$. Choose $Q_P=I$ when $h_P\ge0$; otherwise, choose $Q_P\in\{X,Y,Z\}\setminus\{P\}$, so that $Q_PPQ_P=-P$. By Definition~\ref{def:pauli_stick}, every admissible stick satisfies
\[
\|Q_PS_PQ_P-e^{-ih_PP}\|
\le\eta_0\epsilon_P+|\epsilon_P-|h_P||
\le\eta_0\epsilon_P+r^2.
\]
The three ideal exponentials satisfy
\[
e^{-ih_XX}e^{-ih_YY}e^{-ih_ZZ}
=I-i(h_XX+h_YY+h_ZZ)+O(r^2)
=E+O(r^2).
\]
Consequently,
\[
\begin{aligned}
\|(Q_XS_XQ_X)(Q_YS_YQ_Y)(Q_ZS_ZQ_Z)-E\|
&\le\eta_0(\epsilon_X+\epsilon_Y+\epsilon_Z)+O(r^2)\\
&\le6\eta_0r+O(r^2).
\end{aligned}
\]
Choosing $\eta_0\le1/24$ and then $\epsilon_0$ sufficiently small makes this at most $r/2$. Computing $\log E$ and the three pairs $(\epsilon_P,Q_P)$ takes $O(1)$ time.
\end{proof}

\begin{figure}[H]
    \centering
    \includegraphics[width=1\linewidth]{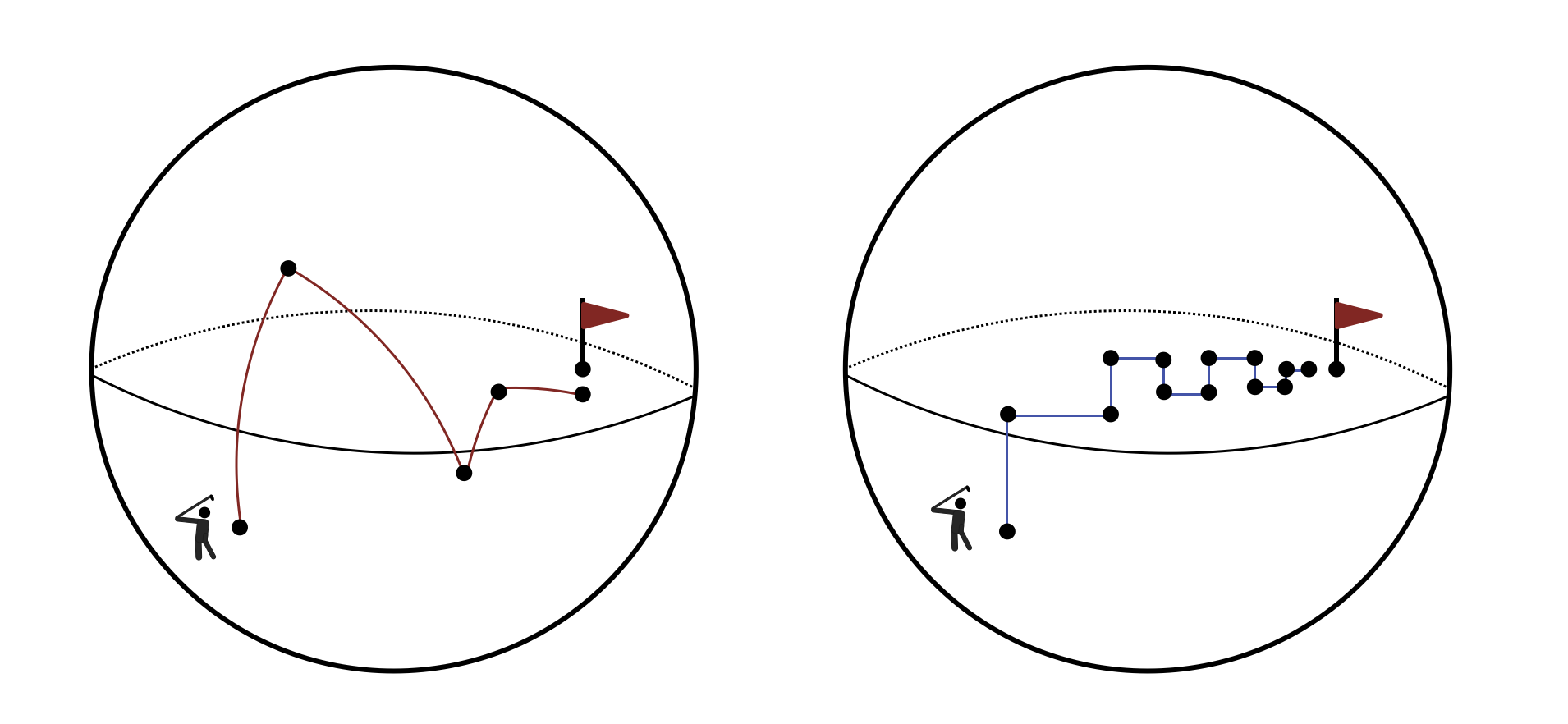}
    \caption{Illustration of (Left) doubly-exponential convergence of Kuperberg's algorithm \cite{kuperberg2023breaking} vs. (Right) singly-exponential convergence of our Pauli Golf method.}
    \label{fig:Golf}
\end{figure}

\subsection{Pauli sticks generation}\label{sec:stick-gen}
Here we show how to construct the Pauli sticks required by Lemma~\ref{lem:pauli_golf} using sequences of Pauli gates and copies of a unitary close to $I$. Proposition~\ref{prop:pauli_stick_generation} bounds the length of these sequences and specifies how accurately this unitary must be approximated.

\begin{proposition}[Pauli stick generation]\label{prop:pauli_stick_generation}
Fix $0<\eta<1/4$. There exist constants $c>0$, $0<\epsilon_0<1$, and $q<1.138$ such that, for every $0<\epsilon\le\epsilon_0$, $\sigma\in\{-1,1\}$, and $P\in\{X,Y,Z\}$, there exists $V\in\SU(2)$ and a word $w_{\mathrm{stick}}(V,X,Y,Z)$ with $\len(w_{\mathrm{stick}})=O(\log^q(1/\epsilon))$ and the following property: for every $\widetilde V\in\SU(2)$ satisfying $\|\widetilde V-V\|\le c/\log^q(1/\epsilon)$, $W_{\mathrm{stick}}(\widetilde V,X,Y,Z))$ is a $(\sigma,\epsilon,P,\eta)$-stick. Moreover, $w_{\mathrm{stick}}$ and the classical description of $V$ can be computed in $\poly\log(1/\epsilon)$ time.
\end{proposition}

\subsubsection{A self-correcting sequence}

Our construction of Pauli sticks uses \emph{self-correction}: we combine repeated copies of a unitary $U$ close to $I$ with Pauli gates to obtain a product whose distance from $I$ is smaller than that of $U$. To compare the efficiency of different self-correcting words, we define their cancellation order and exponent as follows.

\begin{definition}[Cancellation order and exponent]\label{def:cancellation}
Let $w$ be a word consisting of $U,X,Y,Z$ that satisfies $W=I$ when $U=I$. Its \emph{cancellation order}, $\operatorname{can}(w)$, is defined as the largest integer $b$ for which there exists a constant $C>0$ such that we have, for every $0<\epsilon<1$,
\[
    \sup_{\substack{U\in\mathrm{SU}(2)\\
                    \|U-I\|\le\epsilon}}
    \|W(U,X,Y,Z)-I\|
    \le C\epsilon^b.
\]
For any word $w$ with $\operatorname{can}(w)>1$, its \emph{cancellation exponent} is defined as
\[
    \alpha(w)\coloneqq \frac{\log\operatorname{len}(w)}{\log\operatorname{can}(w)}.
\]
\end{definition}
Intuitively, a smaller cancellation exponent means fewer gates for a given cancellation order. We construct our self-correcting sequence by repeatedly applying $w\mapsto(PwP)w$, choosing a Pauli gate $P$ at each step. For example, conjugation by $Z$ reverses the signs of the $X$ and $Y$ coefficients of $W$. In the product $(ZW Z)W$, the first-order $X$ and $Y$ terms from $ZW Z$ cancel those from $W$. Formally, we define the word $w_n^P=w_n^P(U,X,Y,Z)$ for $P\in\{X,Y,Z\}$ by $w_1^X=w_1^Y=w_1^Z=U$ and, for $n\ge1$,
\begin{equation}\label{eq:filters_f(U)}
\begin{aligned}
w_{n+1}^X(U,X,Y,Z)&=
\begin{cases}
Xw_n^Z(U,X,Y,Z) Xw_n^Z(U,X,Y,Z),&n\text{ odd},\\
Xw_n^Y(U,X,Y,Z) Xw_n^Y(U,X,Y,Z),&n\text{ even},
\end{cases}\\
w_{n+1}^Y(U,X,Y,Z)&=
\begin{cases}
Yw_n^X(U,X,Y,Z) Yw_n^X(U,X,Y,Z),&n\text{ odd},\\
Yw_n^Z(U,X,Y,Z) Yw_n^Z(U,X,Y,Z),&n\text{ even},
\end{cases}\\
w_{n+1}^Z(U,X,Y,Z)&=
\begin{cases}
Zw_n^Y(U,X,Y,Z) Zw_n^Y(U,X,Y,Z),&n\text{ odd},\\
Zw_n^X(U,X,Y,Z) Zw_n^X(U,X,Y,Z),&n\text{ even}.
\end{cases}
\end{aligned}
\end{equation}
Under our convention in Section~\ref{sec:preliminaries}, the unitaries $W_n^P=\operatorname{val}(w_n^P)$ satisfy the same recursion as the words $w_n^P$. Note that $W^Y_3 = XUXYUYZUZU = f(U)$ as used in Sardharwalla et al.’s algorithm \cite{sardharwalla2016universal}.

As we will prove later in Lemma~\ref{lem:self_correcting_orders}, the cancellation orders of $w_n^P$ form the \textit{Tribonacci sequence} $1,1,2,4,7,13,\ldots$. For comparison, Kuperberg's construction uses words with Fibonacci cancellation orders and lengths $\Theta(2^n)$~\cite[Section~8]{kuperberg2023breaking}. Our words also have lengths $\Theta(2^n)$, so the faster growth of the Tribonacci numbers gives a smaller cancellation exponent.

\begin{definition}[Tribonacci sequence]
We define the \emph{Tribonacci sequence} $\{T_n\}_{n\ge1}$ by
\[
T_1=T_2=1,\qquad T_3=2,\qquad
T_{n+3}=T_{n+2}+T_{n+1}+T_n
\quad(n\ge1).
\]
Moreover, we write $\rho\in(1,2)$ for the real root of its characteristic function $\rho^3-\rho^2-\rho-1=0$.
\end{definition}

\begin{lemma}
\label{lem:self_correcting_orders}
Let
$q=\log2/\log\rho$. The words $w_n^P$ in \eqref{eq:filters_f(U)} satisfy, for every $n\ge1$ and $P\in\{X,Y,Z\}$,
\begin{equation}\label{eq:self_correcting_orders}
\operatorname{can}(w_n^P)=T_n,
\qquad
\len(w_n^P)
=3\cdot2^{n-1}-2
\le6T_n^q.
\end{equation}
Moreover, its cancellation exponent satisfies
\[
\lim_{n\to\infty}\alpha(w_n^P)
=\frac{\log2}{\log\rho}
\approx1.13747.
\]
\end{lemma}

We utilize the following two facts to prove Lemma~\ref{lem:self_correcting_orders}.

\begin{fact}\label{fact:tribonacci_growth}
For every $n\ge1$, the Tribonacci number $T_n$ satisfies
\[
\rho^{n-2}\le T_n\le2\rho^{n-1}.
\]
\end{fact}

\begin{proof}
Since $1<\rho<2$, both bounds hold for $n=1,2,3$.
The identity $\rho^3=\rho^2+\rho+1$ shows that
the sequences $\rho^{n-2}$ and $2\rho^{n-1}$
satisfy the same recurrence as $T_n$.
The bounds therefore hold for every $n\ge1$
by induction.
\end{proof} 

\begin{fact}\label{fact:self_correcting_step}
For $V=aI-i(xX+yY+zZ)\in\SU(2)$ with $a,x,y,z\in\mathbb R$,
\begin{equation}\label{eq:self_correcting_coefficients}
\begin{aligned}
XVXV&=(1-2x^2)I-i(2axX-2xzY+2xyZ),\\
YVYV&=(1-2y^2)I-i(2yzX+2ayY-2xyZ),\\
ZVZV&=(1-2z^2)I-i(-2yzX+2xzY+2azZ).
\end{aligned}
\end{equation}
In particular, $\|XVXV-I\|=2|x|$, $\|YVYV-I\|=2|y|$, and $\|ZVZV-I\|=2|z|$.
\end{fact}

\begin{proof}[Proof of Lemma~\ref{lem:self_correcting_orders}]
We prove the claims for $w_n^Z$. The proofs for $w_n^X$ and $w_n^Y$ proceeds similarly. The three words have the same length, with $\operatorname{len}(w_1^Z)=1$ and $\operatorname{len}(w_{n+1}^Z)=2\operatorname{len}(w_n^Z)+2$, so $\operatorname{len}(w_n^Z)=3\cdot2^{n-1}-2$. Since $\rho^q=2$, Fact~\ref{fact:tribonacci_growth}
gives $2^{n-1}=2(\rho^{n-2})^q\le2T_n^q.$ Hence
\begin{equation}\label{eq:length_bound}
    \len(w_n^Z)\le6T_n^q
\end{equation}
To determine the cancellation order, we track how each step changes the Pauli coefficients using Fact~\ref{fact:self_correcting_step}. For each $n\ge1$, write
\[
W_n^Z=\tau_nI-i\bigl(v_n^{X}X+v_n^{Y}Y+v_n^{Z}Z\bigr),
\]
where $\tau_n,v_n^{X},v_n^{Y},v_n^{Z}\in\mathbb R$.
For $P\in\{X,Y,Z\}$, let $c_n^{P}$ be the smallest total degree of a term with nonzero coefficient when $v_n^{P}$ is expanded in the initial coefficients $v_1^{X},v_1^{Y},v_1^{Z}$. For example, $v_2^{X}=-2v_1^{Y}v_1^{Z}$ gives $c_2^{X}=2$. Since $W_n^Z=I$ at $U=I$, the expansion of each $\tau_n$ starts with $1$. Cyclically permuting $X,Y,Z$ also permutes the initial coefficients, preserving total degree. Thus the corresponding degree triples for $W_n^X$ and $W_n^Y$ are $(c_n^{Z},c_n^{X},c_n^{Y})$ and $(c_n^{Y},c_n^{Z},c_n^{X})$, respectively. Applying Fact~\ref{fact:self_correcting_step} to $W_{n+1}^Z=ZW_n^Y ZW_n^Y$ for odd $n$ and $W_{n+1}^Z=ZW_n^X ZW_n^X$ for even $n$ gives
\begin{equation}
\label{eq:Fz_cancellation}
\bigl(c_{n+1}^{X},c_{n+1}^{Y},c_{n+1}^{Z}\bigr)=
\begin{cases}
\bigl(c_n^{X}+c_n^{Z},c_n^{X}+c_n^{Y},c_n^{X}\bigr),
&n\text{ odd},\\
\bigl(c_n^{X}+c_n^{Y},c_n^{Y}+c_n^{Z},c_n^{Y}\bigr),
&n\text{ even}.
\end{cases}
\end{equation}
We claim that, for $n\ge2$,
\begin{align}\label{eqn:cn_claimed}
\bigl(c_n^{X},c_n^{Y},c_n^{Z}\bigr)=
\begin{cases}
(T_{n+1},\,T_n+T_{n-1},\,T_n),&n\text{ odd},\\
(T_n+T_{n-1},\,T_{n+1},\,T_n),&n\text{ even}.
\end{cases}
\end{align}
Indeed, for $n=2$, Eqs.\eqref{eq:Fz_cancellation} and~\eqref{eqn:cn_claimed} both give $(2,2,1)$. Substituting either case of \eqref{eqn:cn_claimed} into \eqref{eq:Fz_cancellation} and using $T_{n+2}=T_{n+1}+T_n+T_{n-1}$ gives \eqref{eqn:cn_claimed} at $n+1$. Fix $n\ge2$. Since $W_n^Z=I$ at $U=I$, continuity gives $\tau_n\ge0$ for all $U$ sufficiently close to $I$, and
\begin{equation}\label{eq:self_correcting_norm}
\|W_n^Z-I\|^2
=\frac{2\bigl((v_n^{X})^2+(v_n^{Y})^2+(v_n^{Z})^2\bigr)}
{1+\tau_n}.
\end{equation}
Fact~\ref{fact:self_correcting_step} gives $|v_n^{P}|=O(\|U-I\|^{c_n^{P}})$ by induction from $|v_1^{P}|\le\|U-I\|$ and $|\tau_k|\le1$ for every $k\ge1$, uniformly in $U$ for fixed $n$. Since $\min_P c_n^{P}=c_n^{Z}=T_n$, Eq.~\eqref{eq:self_correcting_norm} gives $\|W_n^Z-I\|=O(\|U-I\|^{T_n})$. To show that the order cannot be larger, take
\begin{equation}\label{eq:pauli_seed_family}
V_r\coloneqq\sqrt{1-3r^2}\,I-ir(X+Y+Z),\qquad 0<r\le1/32.
\end{equation}
Then, induction gives $v_n^{Z}=ar^{T_n}+O(r^{T_n+1})$ for some $a\ne0$. Hence, for sufficiently small $r>0$,
\begin{equation}\label{eq:self_correcting_sharpness}
\|W_n^Z(V_r,X,Y,Z)-I\|\ge\frac{|a|}{2}r^{T_n},
\qquad
\|V_r-I\|\le3r.
\end{equation}
Finally, by Fact~\ref{fact:tribonacci_growth} we have $\log T_n=n\log\rho+O(1)$.
Then by Eq.~\eqref{eq:self_correcting_orders}, we can conclude that
\[
\lim_{n\to\infty}\alpha(w_n^Z)
=\lim_{n\to\infty}
\frac{n\log2+O(1)}{n\log\rho+O(1)}
=\frac{\log2}{\log\rho}
\approx1.13747.
\]
\end{proof}

The next lemma bounds
the $Y$ and $Z$ coefficients of $W_n^X$
relative to its $X$ coefficient.

\begin{lemma}\label{lem:pauli_stick_coefficients}
Let $U\in\SU(2)$ satisfy $\|U-I\|\le1/4$, and let $w_n^X$ be the word in \eqref{eq:filters_f(U)}. For every $n\ge2$, write $W_n^X=aI-i(v_n^{X}X+v_n^{Y}Y+v_n^{Z}Z)$. Then
\[
|v_n^{Y}|+|v_n^{Z}|\le4\|U-I\|\,|v_n^{X}|.
\]
The coefficients of $W_n^Y$ and $W_n^Z$ satisfy the same bound with $(v_n^{X},v_n^{Y},v_n^{Z})$ replaced by $(v_n^{Y},v_n^{Z},v_n^{X})$ and $(v_n^{Z},v_n^{X},v_n^{Y})$, respectively.
\end{lemma}

\begin{proof}
It suffices to consider products obtained by alternating $Z$ and $Y$, starting with $Z$; cyclic permutation of the Pauli axes gives the result for every $P$. For $V=aI-i(xX+yY+zZ)\in\SU(2)$, Fact~\ref{fact:self_correcting_step} gives
\[
\|ZVZV-I\|=2|z|\le2\|V-I\|,
\qquad
\|Y(ZVZV)Y(ZVZV)-I\|=4|xz|\le2\|V-I\|^2.
\]
Induction therefore gives $\|V-I\|\le2\|U-I\|\le1/2$ for every intermediate product $V$. Consequently, $a=1-\|V-I\|^2/2\ge7/8$ and $\sqrt{x^2+y^2+z^2}\le2\|U-I\|$. Thus,
\begin{equation}\label{eq:pauli_stick_transverse_bounds}
\begin{aligned}
2|z|(|x|+|y|)&\le4\|U-I\|\,|2az|,\\
2|y|(|x|+|z|)&\le4\|U-I\|\,|2ay|.
\end{aligned}
\end{equation}
Eq.~\eqref{eq:pauli_stick_transverse_bounds} gives the stated bounds for $ZVZV$ and $YVYV$, respectively.
\end{proof}

\subsubsection{Sticks of sufficiently small size}
\label{sec:stick-size}

We now construct Pauli sticks of any size using the sequences $w_n^P$ in \eqref{eq:filters_f(U)}. We choose their inputs from the family
\[
V_r\coloneqq\sqrt{1-3r^2}\,I-ir(X+Y+Z),\qquad 0<r\le1/32,
\]
which has also been used in the proof of Lemma~\ref{lem:self_correcting_orders}.
The first lemma chooses $n$ and $r$ so that $w_n^P(V_r,X,Y,Z)$ is a $(1,\epsilon,P,\eta)$-stick and bounds the length of the sequence. We write $q=\log2/\log\rho<1.138$.

\begin{lemma}
\label{lem:seed_approximation_accuracy}
Fix $0<\eta<1/4$. For every sufficiently small $\epsilon>0$ and every $P\in\{X,Y,Z\}$, there exist $n\ge2$ and $r\in(0,1/32]$ such that the word $w_n^P$ defined in \eqref{eq:filters_f(U)} is a $(1,\epsilon,P,\eta)$-stick, with
\[
r=\Omega(1),
\qquad
\operatorname{len}(w_n^P(V_r,X,Y,Z))=O\!\left(\log^q(1/\epsilon)\right).
\]
Moreover, $n$ and $r$ can be computed in $\poly\log(1/\epsilon)$ time.
\end{lemma}

\begin{proof}
Fix $0<r\le1/32$. Throughout this proof, for any $j$ we write $w_j^P(V_r)= w_j^P(V_r,X,Y,Z)$ and $W_j^P(V_r)= W_j^P(V_r,X,Y,Z)$ for simplicity. By cyclic symmetry, the distance $d_j(r)=\|W_j^P(V_r)-I\|$ is independent of $P$. For each $j\ge1$, define
$\tau_j,v_j^X,v_j^Y,v_j^Z\in\mathbb R$ by
\[
\begin{aligned}
W_j^Y(V_r)
&=\tau_j I-i(v_j^X X+v_j^Y Y+v_j^Z Z),
\qquad j\text{ odd},\\
W_j^Z(V_r)
&=\tau_j I-i(v_j^X X+v_j^Y Y+v_j^Z Z),
\qquad j\text{ even}.
\end{aligned}
\]
Applying Fact~\ref{fact:self_correcting_step} twice gives $d_{j+2}(r)\le2d_j(r)^2$ for every $j\ge1$. Since $d_1(r)=\|V_r-I\|\le2r$ and $d_2(r)=2r$, induction using $2(4r)^2\le4r$ gives $d_j(r)\le4r$ for every $j\ge1$. Consequently,
\[
\tau_j=1-\frac{d_j(r)^2}{2}
\ge1-8r^2\ge\frac78.
\]
We first check that the dominant coefficient has the required sign. Starting from $v_1^X=v_1^Y=v_1^Z=r>0$, Fact~\ref{fact:self_correcting_step} gives 
\[
(+,+,+)\xrightarrow{Z}(-,+,+)\xrightarrow{Y}(+,+,+),
\]
where the triples record the signs of
$v_j^X,v_j^Y,v_j^Z$; for example, $(-,+,+)$ means
that only the $X$-coordinate is negative, and an arrow labelled $P$ denotes the update
$w\mapsto PwPw$. Since the scalar coefficients remain positive, this pattern repeats. The $Z$-coordinate is positive after each $Z$-step, and the $Y$-coordinate is positive after each $Y$-step. Cyclically relabeling $X\mapsto Y\mapsto Z\mapsto X$ leaves $V_r$ unchanged and permutes the three word families, carrying their coefficients unchanged to the corresponding axes.
Thus the coefficient of $-iP$ in $W_j^P(V_r)$ is positive for every $j\ge2$ and every $P$, as required for the $(1,\epsilon,P,\eta)$-stick.

We next obtain more precise bounds on the size.
For odd $j$, applying Fact~\ref{fact:self_correcting_step} to the four next steps $Z,Y,Z,Y$ in \eqref{eq:filters_f(U)} gives
\begin{equation}\label{eq:stick_distance_steps}
\begin{aligned}
d_{j+1}(r)&=2|v_j^Z|, & d_{j+2}(r)&=4|v_j^Xv_j^Z|,\\
d_{j+3}(r)&=16|v_j^Xv_j^Y(v_j^Z)^2|, & d_{j+4}(r)&=256|\tau_j(v_j^X)^2v_j^Y(v_j^Z)^4|.
\end{aligned}
\end{equation}
For even $j$, the steps are $Y,Z,Y,Z$, and the same formulas hold with $v_j^Y$ and $v_j^Z$ interchanged. The fourth expression in \eqref{eq:stick_distance_steps} is $2\tau_j$
times the product of the first three. Substituting $2\tau_j=2-d_j(r)^2$ gives
\begin{equation}\label{eq:stick_distance_recursion}
d_{j+4}(r)=\bigl(2-d_j(r)^2\bigr)d_{j+1}(r)d_{j+2}(r)d_{j+3}(r)
\end{equation}
for any $j\geq 1$. At $j=1$, the coefficients $v_1^X=v_1^Y=v_1^Z=r$ give $d_2(r)=2r$, $d_3(r)=4r^2$, and $d_4(r)=16r^4$. Since $1\le2-d_j(r)^2\le2$, induction using the Tribonacci recurrence gives
\begin{equation}\label{eq:stick_norm_bounds}
2r^{T_j}\le d_j(r)\le\frac{(4r)^{T_j}}2,\qquad\forall j\geq 2.
\end{equation}
We now choose $n$ and $r$ so that
$d_n(r)=2\sin(\epsilon/2)$, the distance from
$e^{-i\epsilon P}$ to $I$.
Fix $0<r_0\le\eta/64$, depending only on $\eta$,
and assume $0<\epsilon\le r_0$.
We first choose a depth for which the output at $r_0$
is larger than the target. Since $T_2=1$ and
$T_j\to\infty$, we can choose $n\ge2$ such that
\[
r_0^{T_{n+1}}<\epsilon\le r_0^{T_n}.
\]
We keep this depth fixed and decrease $r$.
The bound $T_{n+1}\le2T_n$ explains why squaring $r_0$
is sufficient to cross the gap between these scales:
$r_0^{2T_n}\le r_0^{T_{n+1}}$.
Taking the smaller endpoint to be $r_0^2/8$ also
accounts for the constants in \eqref{eq:stick_norm_bounds}, which gives
\[
d_n(r_0^2/8)
\le\frac14 r_0^{2T_n}
\le\frac14 r_0^{T_{n+1}}
<\frac{\epsilon}{4},\qquad
d_n(r_0)
\ge2r_0^{T_n}
\ge2\epsilon.
\]
Thus the two endpoint outputs lie on opposite sides
of the target distance, since
$\epsilon/2\le2\sin(\epsilon/2)\le\epsilon$.
By continuity, there is an
$r_*\in[r_0^2/8,r_0]$ such that
\begin{equation}\label{eq:stick_target_distance}
d_n(r_*)=2\sin(\epsilon/2).
\end{equation}
In particular, the same fixed interval of input
parameters works for every sufficiently small
$\epsilon$, and $r_*\ge r_0^2/8=\Omega(1)$.

Having fixed the size, we now control the direction of $W_n^P$.
At $r=r_*$, we have $\tau_n=\cos\epsilon$ and
$\sum_Q(v_n^Q)^2=\sin^2\epsilon$.
After cyclically relabeling the coefficients for the
chosen $P$, positivity gives
$0\le\sin\epsilon-v_n^P\le\sum_{Q\ne P}|v_n^Q|$.
Since the scalar coefficients already agree,
Lemma~\ref{lem:pauli_stick_coefficients}, applied with
$\|U_{r_*}-I\|\le2r_*$, yields
\[
\|W_n^P(U_{r_*})-e^{-i\epsilon P}\|
\le2\sum_{Q\ne P}|v_n^Q|
\le16r_*\sin\epsilon
\le\frac{\eta\epsilon}{4}.
\]
Thus $W_n^P(U_{r_*})$ has the required size
and direction.

For the length bound, our choice of $n$ gives
$T_n\le\log(1/\epsilon)/\log(1/r_0)
=O(\log(1/\epsilon))$.
Using Fact~\ref{fact:tribonacci_growth}, we obtain
\begin{equation}\label{eq:uniform_stick_length}
\operatorname{len}(w_n^P(V_r))
=3\cdot2^{n-1}-2
\le6T_n^q
=O\left(\log^q(1/\epsilon)\right).
\end{equation}
Finally, we show how to compute $n$ and $r$ efficiently. We compute $n$ by generating the Tribonacci numbers until $r_0^{T_{n+1}}<\epsilon\le r_0^{T_n}$, which can be done in time $O(n)$. We then find $r$ by bisection, starting with
$a=r_0^2/8$ and $b=r_0$.
Throughout the search, we keep
\[
d_n(a)<2\sin(\epsilon/2)<d_n(b).
\]
The word contains $2^{n-1}$ copies of its input,
and $\|V_r-V_t\|\le2|r-t|$, so
\begin{equation}\label{eq:stick_parameter_lipschitz}
|d_n(r)-d_n(t)|
\le\|W_n^P(V_r)-W_n^P(V_t)\|
\le2^n|r-t|.
\end{equation}
At each step, set $r=(a+b)/2$ and evaluate
$d_n(r)-2\sin(\epsilon/2)$ with error at most
$\eta\epsilon/64$.
If the absolute value of the computed result is
at most $3\eta\epsilon/64$, stop.
The evaluation error then guarantees
\begin{equation}\label{eq:stick_calibration_accuracy}
|d_n(r)-2\sin(\epsilon/2)|
\le\frac{\eta\epsilon}{16}.
\end{equation}
Otherwise, the computed result has the correct sign:
if it is negative, replace $a$ by $r$;
if it is positive, replace $b$ by $r$.
The endpoint inequalities are preserved, so the
interval always contains a parameter satisfying
Eq.~\eqref{eq:stick_target_distance}.
Once $2^n(b-a)\le\eta\epsilon/32$,
Eq.~\eqref{eq:stick_parameter_lipschitz} ensures that
the next evaluation passes the stopping test.
Since the interval width halves at each step,
the search takes $O(n+\log(1/\epsilon))$ steps.

We now check that the returned parameter gives
the required stick.
The angle $\theta=2\arcsin(d_n(r)/2)$ satisfies
$|\theta-\epsilon|\le\eta\epsilon/8$, since
$x\mapsto2\arcsin(x/2)$ is $2$-Lipschitz on $[0,1]$.
Applying the earlier direction estimate at this
angle gives
\[
\|W_n^P(V_r)-e^{-i\epsilon P}\|
\le16r\sin\theta+|\theta-\epsilon|
<\eta\epsilon,
\]
where we used
$\sin\theta\le d_n(r)\le2\epsilon$
and $r\le r_0\le\eta/64$.
The returned parameter remains in
$[r_0^2/8,r_0]$, so $r=\Omega(1)$.

We conclude the proof by discussing the time complexity of this procedure. The length bound in \eqref{eq:uniform_stick_length} gives $n=O(\log\log(1/\epsilon))$.
Thus the search takes $O(\log(1/\epsilon))$ steps,
each using $O(n)$ coefficient updates.
The total computation therefore takes
$\poly\log(1/\epsilon)$ time.
\end{proof}
The preceding construction uses the ideal input $V_r$, which we will approximate by a word over our gate set. The next lemma bounds the resulting error relative to the size of the stick.
\begin{lemma}\label{lem:pauli_stick_input_error}
For every integer $n\ge2$, $P\in\{X,Y,Z\}$, $0<r\le1/32$, and $\widetilde V\in\SU(2)$ satisfying $\|\widetilde V-V_r\|\le r/2^{n+2}$, the word $w_n^P$ defined in \eqref{eq:filters_f(U)} satisfies
\begin{equation}\label{eq:pauli_stick_input_error}
\|W_n^P(\widetilde V,X,Y,Z)-W_n^P(V_r,X,Y,Z)\|
\le\frac{2^{n+2}}{r}\,\|\widetilde V-V_r\|\,\|W_n^P(V_r,X,Y,Z)-I\|.
\end{equation}
\end{lemma}

\begin{proof}
Throughout this proof, for any $j\geq 1$ and $V\in\SU(2)$ we write $w_j^P(V)= w_j^P(V,X,Y,Z)$ and $W_j^P(V_r)= W_j^P(V,X,Y,Z)$ for simplicity. It suffices to prove \eqref{eq:pauli_stick_input_error} with $P=Y$ for odd $n$ and $P=Z$ for even $n$. Cyclically relabeling $X\mapsto Y\mapsto Z\mapsto X$ gives the remaining cases, as it leaves $V_r$ unchanged and preserves both the recursion and the operator norm. For odd $k$, write
\[
\begin{aligned}
W_k^Y(V_r)
&=\tau_k I-i(v_k^X X+v_k^Y Y+v_k^Z Z),\\
W_k^Y(\widetilde V)
&=\widetilde\tau_k I
-i(\widetilde v_k^X X+\widetilde v_k^Y Y
+\widetilde v_k^Z Z).
\end{aligned}
\]
For even $k$, use the same notation for
$W_k^Z(V_r)$ and $W_k^Z(\widetilde V)$.
Throughout this proof, sums over $Q$ are taken over $\{X,Y,Z\}$.
We first show that both sequences are close to $I$.
Since $\|V_r-I\|\le2r$ and
$\|\widetilde V-V_r\|\le r/2$, both inputs lie within
$5r/2$ of $I$.
The uniform bound in the proof of
Lemma~\ref{lem:pauli_stick_coefficients} therefore gives,
at every depth and for every Pauli axis $P$,
\begin{equation}\label{eq:pauli_input_iterate_norms}
\max\bigl\{
\|W_k^P(V_r)-I\|,
\|W_k^P(\widetilde V)-I\|
\bigr\}
\le5r\le\frac14.
\end{equation}
Consequently, $\tau_k,\widetilde\tau_k\ge1/2$ and
\[
\sum_Q(v_k^Q)^2\le\frac1{16},
\qquad
\sum_Q(\widetilde v_k^Q)^2\le\frac1{16}.
\]
Starting from $v_1^Q=r>0$,
Eq.~\eqref{eq:self_correcting_coefficients} also shows
that every ideal Pauli coefficient remains nonzero,
so the relative error
$|\widetilde v_k^Q-v_k^Q|/|v_k^Q|$
is well-defined.

Suppose that
$|\widetilde v_k^Q-v_k^Q|\le\delta|v_k^Q|$
for every $Q$, where $0\le\delta\le1$.
Normalization gives the same relative bound for the
scalar coefficient:
\begin{equation}\label{eq:pauli_input_scalar_error}
\begin{aligned}
|\widetilde\tau_k-\tau_k|
=
\frac{
\left|\sum_Q\bigl((\widetilde v_k^Q)^2-(v_k^Q)^2\bigr)\right|
}{
\widetilde\tau_k+\tau_k
}
\le3\delta\sum_Q(v_k^Q)^2
\le\delta\tau_k.
\end{aligned}
\end{equation}
Each new Pauli coefficient in \eqref{eq:self_correcting_coefficients} is a product
of two previous coefficients, multiplied by $\pm2$.
Relative error at most $\delta$ in each factor
therefore gives relative error at most
$(1+\delta)^2-1$:
\[
|\widetilde v_{k+1}^Q-v_{k+1}^Q|
\le(2\delta+\delta^2)|v_{k+1}^Q|.
\]
Initially, $v_1^Q=r$ and
$|\widetilde v_1^Q-r|\le\|\widetilde V-V_r\|$.
Iterating this estimate gives, for every $Q$ and
$1\le k\le n$,
\begin{equation}\label{eq:pauli_input_relative_bound}
\frac{|\widetilde v_k^Q-v_k^Q|}{|v_k^Q|}
\le
\left(1+\frac{\|\widetilde V-V_r\|}{r}\right)^{2^{k-1}}-1
\le\frac{2^k}{r}\,\|\widetilde V-V_r\|
\le\frac14,
\end{equation}
where the hypothesis gives
$2^{n-1}\|\widetilde V-V_r\|/r\le1/8$.
In particular, the relative errors remain below $1$,
as required to apply the one-step estimate throughout
the induction.

We now convert these relative coefficient bounds into an operator-norm bound. The Euclidean norm of the ideal Pauli coefficients is at most $1/4$ and no larger than $\|W_n^P(V_r)-I\|$.
Thus Eqs.~\eqref{eq:pauli_input_scalar_error} and \eqref{eq:pauli_input_relative_bound} give
\[
\begin{aligned}
\|W_n^P(\widetilde V)-W_n^P(V_r)\|
&\le
\frac{2^n}{r}\,\|\widetilde V-V_r\|
\left(
\sqrt{\sum_Q(v_n^Q)^2}
+3\sum_Q(v_n^Q)^2
\right)\\
&\le
\frac{2^{n+1}}{r}\,\|\widetilde V-V_r\|\,
\|W_n^P(V_r)-I\|.
\end{aligned}
\]
\end{proof}

Now we prove Proposition~\ref{prop:pauli_stick_generation}.
\begin{proof}[Proof of Proposition~\ref{prop:pauli_stick_generation}]
Take $q=\log2/\log\rho<1.138$ and choose $\epsilon_0>0$ to be sufficiently small. First consider $\sigma=1$. For every $0<\epsilon\le\epsilon_0$, apply Lemma~\ref{lem:seed_approximation_accuracy} with tolerance $\eta/2$, reserving the remaining $\eta\epsilon/2$ for approximating the input unitary $V_r$ of the word $w_n^P$. This gives $n\ge2$ and $r=\Omega(1)$ such that
\[
\|W_n^P(V_r,X,Y,Z)-e^{-i\epsilon P}\|
\le\frac{\eta\epsilon}{2},
\qquad
\len(w_n^P)
=O\!\left(\log^q(1/\epsilon)\right).
\]
Take $V=V_r$ and $w_{\mathrm{stick}}=w_n^P$. The triangle inequality also gives $\|W_{\mathrm{stick}}(V,X,Y,Z)-I\|\le2\epsilon$. Since $r=\Omega(1)$ and $2^{n-1}\le\len(w_{\mathrm{stick}})$, there is a constant $C_\eta>0$ such that
\[
\frac{2^{n+2}}{r}
\le C_\eta\log^q(1/\epsilon).
\]
Fix $c=\eta/(4C_\eta)$.
For every $\widetilde V\in\SU(2)$ satisfying
$\|\widetilde V-V\|\le c/\log^q(1/\epsilon)$, we have
\[
\frac{2^{n+2}}{r}\,\|\widetilde V-V\|
\le C_\eta c=\frac{\eta}{4}<1.
\]
Thus Lemma~\ref{lem:pauli_stick_input_error} applies
and gives
\[
\|W_{\mathrm{stick}}(\widetilde V,X,Y,Z)-W_{\mathrm{stick}}(V,X,Y,Z)\|
\le\frac{\eta}{4}\|W_{\mathrm{stick}}(V,X,Y,Z)-I\|
\le\frac{\eta\epsilon}{2}.
\]
Combining the two error bounds shows that $W_{\mathrm{stick}}(\widetilde V,X,Y,Z)$ is a $(1,\epsilon,P,\eta)$-stick.

For $\sigma=-1$, replace $w_{\mathrm{stick}}$
by $Qw_{\mathrm{stick}}Q$ for any Pauli $Q\ne P$.
Since $QPQ=-P$, conjugation reverses the target angle
and preserves the approximation error.
It adds only two gates, so the length remains
$O(\log^q(1/\epsilon))$.

Finally, using Lemma~\ref{lem:seed_approximation_accuracy}, we can compute $n$ and $r$, and hence a classical description of $V=V_r$, in $\poly\log(1/\epsilon)$ time. We then construct $w_{\mathrm{stick}}$ from the recursion for $w_n^P$ in time proportional to its length, which is $O(\log^q(1/\epsilon))$.
\end{proof}

\subsection{The compilation algorithm}\label{sec:compilation-alg}

We now show how to use gates from $G$ to approximate a target unitary to the desired accuracy. We start with a coarse approximation and improve it by appending a correction formed from three Pauli sticks, reducing the error by at least a factor of two. Algorithms~\ref{alg:compile} and~\ref{alg:axis} implement these steps recursively using Lemma~\ref{lem:pauli_golf} and Proposition~\ref{prop:pauli_stick_generation}, respectively. Throughout, we fix $\eta_0$ as in Lemma~\ref{lem:pauli_golf}, and take $0<c\le1$ and $0<q<1.138$ from Proposition~\ref{prop:pauli_stick_generation} with tolerance $\eta_0$. Fix $\epsilon_0$ to be small enough that is at most $1/8$, the constant $\epsilon_0$ in Lemma~\ref{lem:pauli_golf}, half the constant $\epsilon_0$ in Proposition~\ref{prop:pauli_stick_generation}, and such that any $0<\epsilon\leq \epsilon_0$ satisfies $c/\log^q(1/\epsilon^2)\ge2\epsilon$.

\begin{algorithm}[H]
\caption{\textsc{Compile}$(U,\epsilon)$}\label{alg:compile}
\begin{algorithmic}[1]
\Require $U\in\mathrm{SU}(2)$; $0<\epsilon<1$
\Ensure a $G$-word $w$ with $\|W-U\|\le\epsilon$
\lIf{$\epsilon\ge\epsilon_0$}{\Return $\epsilon_0\text{-}\textsc{ConstantNet}(U)$ (Lemma~\ref{lem:compilation_base_case}).}
\State $w_0\gets\Call{Compile}{U,2\epsilon}$, $E\gets W_0^{-1}U$ 
\lIf{$\|E-I\|\le\epsilon$}{\Return $w_0$}
\State Compute $\epsilon_P,Q_P$ for $P\in\{X,Y,Z\}$ from $E$ using Lemma~\ref{lem:pauli_golf}.
\For{$P\in\{X,Y,Z\}$}
  \State $w_{\mathrm{stick}}^P\gets\Call{PauliStick}{P,\epsilon_P}$
\EndFor
\State \Return $w_0(Q_Xw_{\mathrm{stick}}^XQ_X)(Q_Yw_{\mathrm{stick}}^YQ_Y)(Q_Zw_{\mathrm{stick}}^ZQ_Z)$
\end{algorithmic}
\end{algorithm}

\begin{algorithm}[H]
\caption{\textsc{PauliStick}$(P,\epsilon)$}\label{alg:axis}
\begin{algorithmic}[1]
\Require $P\in\{X,Y,Z\}$; $0<\epsilon\le2\epsilon_0$
\Ensure a $G$-word $w$ with $\|W-e^{-i\epsilon P}\|\le\eta_0\epsilon$
\State Compute $V,w_{\mathrm{stick}}$ from $P,\epsilon$ using Proposition~\ref{prop:pauli_stick_generation}.
\State $w_V\gets\Call{Compile}{V,c/\log^q(1/\epsilon)}$
\State \Return $w_{\mathrm{stick}}(w_V,X,Y,Z)$
\end{algorithmic}
\end{algorithm}

Write $L_{\mathrm{full}}(\epsilon),T_{\mathrm{full}}(\epsilon)$ for the worst-case output length and runtime of Algorithm~\ref{alg:compile} over all target unitaries and $0<\epsilon<1$ and $L_{\mathrm{stick}}(\epsilon),T_{\mathrm{stick}}(\epsilon)$ for those of Algorithm~\ref{alg:axis} over all $P\in\{X,Y,Z\}$ and stick sizes in $[\epsilon,2\epsilon_0]$. For the base case $\epsilon_0\le\epsilon<1$, Lemma~\ref{lem:compilation_base_case} gives $L_{\mathrm{full}}(\epsilon)=O(1)$ and $T_{\mathrm{full}}(\epsilon)=O(1)$. The next two lemmas give recursive bounds on the output lengths and runtime of both algorithms.

\begin{lemma}\label{lem:pauli_compiler_recursion}
Let $U\in\SU(2)$ and $0<\epsilon<1$. Algorithm~\ref{alg:compile} terminates and returns a $G$-word $w$ satisfying $\|W-U\|\le\min\{\epsilon,\epsilon_0\}$. For $0<\epsilon<\epsilon_0$, its length and runtime satisfy
\begin{equation}\label{eq:pauli_full_recursion}
\begin{aligned}
L_{\mathrm{full}}(\epsilon)
&\le L_{\mathrm{full}}(2\epsilon)+3L_{\mathrm{stick}}(\epsilon^2)+6,\\
T_{\mathrm{full}}(\epsilon)
&\le T_{\mathrm{full}}(2\epsilon)+3T_{\mathrm{stick}}(\epsilon^2)
+L_{\mathrm{full}}(\epsilon)\poly\log(1/\epsilon).
\end{aligned}
\end{equation}
\end{lemma}

\begin{proof}
The correctness claim follows from Lemma~\ref{lem:compilation_base_case} for the base case. For each correction, Proposition~\ref{prop:pauli_stick_generation} supplies the required sticks, and Lemma~\ref{lem:pauli_golf} bounds the final error by $\|E-I\|/2\le\epsilon$.

For the sequence length, note that Algorithm~\ref{alg:compile} appends the three words $w_{\mathrm{stick}}^X,w_{\mathrm{stick}}^Y,w_{\mathrm{stick}}^Z$ and at most six Pauli gates to $w_0$. The word $w_0$ has length at most $L_{\mathrm{full}}(2\epsilon)$, and each $w_{\mathrm{stick}}^P$ has length at most $L_{\mathrm{stick}}(\epsilon^2)$ because $\epsilon_P>\epsilon^2$. This gives the claimed length bound.

For the runtime, the call producing $w_0$ and the three calls producing $w_{\mathrm{stick}}^X,w_{\mathrm{stick}}^Y,w_{\mathrm{stick}}^Z$ cost at most $T_{\mathrm{full}}(2\epsilon)+3T_{\mathrm{stick}}(\epsilon^2)$ in total. Computing $\epsilon_P,Q_P$ takes $O(1)$ time by Lemma~\ref{lem:pauli_golf}, while evaluating $W_0$ and writing $w$ take $L_{\mathrm{full}}(\epsilon)\poly\log(1/\epsilon)$ time. Adding these costs gives the claimed runtime bound.
\end{proof}

\begin{lemma}\label{lem:pauli_stick_recursion}
Let $P\in\{X,Y,Z\}$ and $0<\epsilon\le2\epsilon_0$. Algorithm~\ref{alg:axis} terminates and returns a $G$-word $w$ whose value is a $(1,\epsilon,P,\eta_0)$-stick. Its length and runtime satisfy, for some constant $C>0$,
\begin{equation}\label{eq:pauli_stick_recursion}
\begin{aligned}
L_{\mathrm{stick}}(\epsilon)
&\le C\log^q(1/\epsilon)
\left[L_{\mathrm{full}}\!\left(\frac{c}{\log^q(1/\epsilon)}\right)+1\right],\\
T_{\mathrm{stick}}(\epsilon)
&\le T_{\mathrm{full}}\!\left(\frac{c}{\log^q(1/\epsilon)}\right)
+L_{\mathrm{stick}}(\epsilon)\poly\log(1/\epsilon).
\end{aligned}
\end{equation}
\end{lemma}

\begin{proof}
The call producing $w_V$ terminates and gives $\|W_V-V\|\le c/\log^q(1/\epsilon)$ by Lemma~\ref{lem:pauli_compiler_recursion}. The correctness claim of Algorithm~\ref{alg:axis} then follows from Proposition~\ref{prop:pauli_stick_generation}.

For the sequence length, note that Proposition~\ref{prop:pauli_stick_generation} gives $\len(w_{\mathrm{stick}})=O(\log^q(1/\epsilon))$. The output $w_{\mathrm{stick}}(w_V,X,Y,Z)$ replaces each occurrence of $V$ in $w_{\mathrm{stick}}$ by $w_V$, whose length is at most $L_{\mathrm{full}}(c/\log^q(1/\epsilon))$. Counting these copies and the Pauli gates gives the claimed length bound.

For the runtime, the call producing $w_V$ costs at most $T_{\mathrm{full}}(c/\log^q(1/\epsilon))$. Computing $w_{\mathrm{stick}}$ and the description of $V$ takes $\poly\log(1/\epsilon)$ time by Proposition~\ref{prop:pauli_stick_generation}, while writing $w$ takes $L_{\mathrm{stick}}(\epsilon)\poly\log(1/\epsilon)$ time. Adding these costs gives the claimed runtime bound.
\end{proof}

\begin{proof}[Proof of Theorem~\ref{thm:free_pauli_compilation_3}]
The guarantee that $\|W-U\|\le\epsilon$ follows
directly from Lemma~\ref{lem:pauli_compiler_recursion}.
We next prove
\begin{equation}\label{eq:pauli_compiler_length_induction}
L_{\mathrm{full}}(\epsilon)
\le C_{\mathrm{len}}
\bigl(1+\log^{2.138}(1/\epsilon)\bigr),\qquad\forall 0<\epsilon<1
\end{equation}
for some $C_{\mathrm{len}}\ge1$,
independent of $\epsilon$, for which this bound
holds for $0<\epsilon<1$.
We prove the bound on $[2^{-n}\epsilon_0,1)$
by induction on the integer $n\ge0$, initially
allowing a constant $C_n$ that may depend on $n$.
We will show that these constants eventually
stop increasing. Lemma~\ref{lem:compilation_base_case} gives
$L_{\mathrm{full}}(\epsilon)=O(1)$ on
$[\epsilon_0,1)$, so a suitable $C_0\ge1$ exists.

For an integer $n\ge0$, assume
Eq.~\eqref{eq:pauli_compiler_length_induction}
holds with $C_n$ in place of $C_{\mathrm{len}}$
for all $\epsilon\in[2^{-n}\epsilon_0,1)$.
Let $\epsilon\in[2^{-n-1}\epsilon_0,2^{-n}\epsilon_0)$.
Combining the length bounds in Lemmas~\ref{lem:pauli_compiler_recursion} and~\ref{lem:pauli_stick_recursion} over the $O(\log(1/\epsilon))$ calls with error parameters $\epsilon,2\epsilon,4\epsilon,\ldots$ gives
\begin{equation}\label{eq:pauli_compiler_length_recurrence}
L_{\mathrm{full}}(\epsilon)
\le C\log^{q+1}(1/\epsilon)
\left[
L_{\mathrm{full}}\!\left(
\frac{c}{\log^q(1/\epsilon^2)}
\right)+1
\right],
\end{equation}
for some constant $C>0$. By our choice of $\epsilon_0$,
\begin{equation}\label{eq:pauli_compiler_length_threshold}
2^{-n}\epsilon_0
\le2\epsilon
\le\frac{c}{\log^q(1/\epsilon^2)}
<1,
\end{equation}
so applying the induction hypothesis to the recursive
term in Eq.~\eqref{eq:pauli_compiler_length_recurrence}
gives
\begin{equation}\label{eq:pauli_compiler_length_induction_close}
\frac{L_{\mathrm{full}}(\epsilon)}
{1+\log^{2.138}(1/\epsilon)}
\le 
\frac{
C_nC\bigl[
2+\log^{2.138}(\log^q(1/\epsilon^2)/c)
\bigr]
}{
\log^{1.138-q}(1/\epsilon)
}.
\end{equation}
The expression multiplying $C_n$ on the right-hand side
of Eq.~\eqref{eq:pauli_compiler_length_induction_close}
has a finite upper bound on the new interval.
Use this bound to choose $C_{n+1}\ge C_n$ so that
the estimate holds throughout
$[2^{-n-1}\epsilon_0,1)$, taking $C_{n+1}=C_n$
whenever the fraction is at most $1$ throughout
the new interval.
Since $q<1.138$, the fraction tends to zero
as $\epsilon\to0$, so the constants $C_n$
eventually stop increasing.
Taking $C_{\mathrm{len}}$ to be their final value
proves Eq.~\eqref{eq:pauli_compiler_length_induction}
for every $0<\epsilon<1$, and hence
\begin{equation}\label{eq:compiler_length_final}
\len(w)\le L_{\mathrm{full}}(\epsilon)
=O\!\left(1+\log^{2.138}(1/\epsilon)\right).
\end{equation}
For the runtime, Eq.~\eqref{eq:compiler_length_final} and Lemma~\ref{lem:pauli_stick_recursion} give $L_{\mathrm{stick}}(\epsilon^2)=\poly\log(1/\epsilon)$. The $O(\log(1/\epsilon))$ calls with error bounds $\epsilon,2\epsilon,4\epsilon,\ldots$ each make at most three calls to Algorithm~\ref{alg:axis}, with stick sizes at least $\epsilon^2$. Their total runtime is at most $C\log(1/\epsilon)T_{\mathrm{stick}}(\epsilon^2)$. Applying the runtime bound in Lemma~\ref{lem:pauli_stick_recursion} then gives
\begin{equation}\label{eq:pauli_compiler_runtime_recurrence}
\begin{aligned}
T_{\mathrm{full}}(\epsilon)
&\le C\log(1/\epsilon)T_{\mathrm{stick}}(\epsilon^2)
+\poly\log(1/\epsilon)\\
&\le C\log(1/\epsilon)\,
T_{\mathrm{full}}\!\left(\frac{c}{\log^q(1/\epsilon^2)}\right)
+\poly\log(1/\epsilon).
\end{aligned}
\end{equation}
Choose constants $k>1$ and $C_{\mathrm{time}}>0$ so that the last term in \eqref{eq:pauli_compiler_runtime_recurrence} is at most $C_{\mathrm{time}}\log^k(1/\epsilon)/2$. For sufficiently small $\epsilon$, we have $2\epsilon\le c/\log^q(1/\epsilon^2)\le\epsilon_0$ and
\[
C\log(1/\epsilon)\log^k\!\left(\frac{\log^q(1/\epsilon^2)}{c}\right)
\le\frac12\log^k(1/\epsilon),
\]
because $\log(\log^q(1/\epsilon^2)/c)=O(\log\log(1/\epsilon))$ and $k>1$. Therefore, by \eqref{eq:pauli_compiler_runtime_recurrence} we have
\[
T_{\mathrm{full}}(\epsilon)
\le C_{\mathrm{time}}C\log(1/\epsilon)\log^k\!\left(\frac{\log^q(1/\epsilon^2)}{c}\right)
+\frac{C_{\mathrm{time}}}{2}\log^k(1/\epsilon)
\le C_{\mathrm{time}}\log^k(1/\epsilon).
\]
\end{proof}

\section{A subcubic IFSK algorithm}\label{sec:inverse_free}

In this section, we present a subcubic IFSK algorithm for an arbitrary finite universal gate set $G\subset\mathrm{SU}(2)$, and prove Theorem~\ref{thm:general_inverse_free}.

\begin{theorem}[Formal version of Theorem~\ref{thm:general_inverse_free_informal}]\label{thm:general_inverse_free}
Let $G\subset\SU(2)$ be a finite universal gate set. Given $U\in\SU(2)$ and $0<\epsilon<1$, Algorithm~\ref{alg:full_compile} outputs a $G$-word $w$ satisfying $\|W-U\|\le\epsilon$ and $\len(w)=O(\log^{2.988}(1/\epsilon))$ in time $\poly\log(1/\epsilon)$.
\end{theorem}

Our algorithm recursively improves an approximation to the target $U$ by appending a correction that reduces the error from $O(\epsilon^2)$ to $\epsilon^3$ in each round. In Section~\ref{sec:directed_unitaries}, we construct \textit{directed unitaries}, which lie near the identity and point approximately in a prescribed direction. We then use them in Section~\ref{sec:local_correction_net} to approximate any required correction in a small region near the identity. Both constructions need only coarser approximations, which we obtain recursively. Section~\ref{sec:full_compilation} combines these constructions into Algorithm~\ref{alg:full_compile} and analyzes its output length and runtime.

\subsection{Directed unitaries}\label{sec:directed_unitaries}
Here we introduce the concept of a \textit{directed unitary}, which is a unitary whose length (angle of rotation) is specified only up to a multiplicative constant, but whose direction is tightly constrained to lie within a fixed small angle of a desired direction. 

\begin{definition}\label{def:directed_unitary}
Let $\epsilon>0$ and let $e\in\mathbb R^3$ be a unit vector. We say that a unitary $W\in\SU(2)$ is \emph{$(e,\epsilon)$-directed} if $\epsilon\le\|\log W\|\le100\epsilon$ and $\left\|\frac{\log W}{\|\log W\|}-e\right\|\le1/20$.
\end{definition}
The following proposition shows how to construct such a unitary in any desired direction $e\in\R^3$ utilizing the sequence
\begin{align}
\Pthirtytwo(Y_1,Y_2)
=
(Y_1^3Y_2)(Y_1^2Y_2Y_1)(Y_2^3Y_1)(Y_2Y_1Y_2Y_1)(Y_2^3Y_1)(Y_2^2Y_1Y_2)(Y_1^3Y_2)(Y_1Y_2Y_1Y_2).
\label{eq:P32}
\end{align}
where $Y_1,Y_2\in\SU(2)$, allowing errors much larger than $\epsilon$ in the unitaries used to form the sequence.

\begin{proposition}\label{prop:directed_stick}
There exists a constant $\epsilon_0>0$ such that, for every unit vector
$e\in\mathbb R^3$ and every $0<\epsilon<\epsilon_0$, there is
$A\in\SU(2)$ with the following property: for every
$\widetilde A\in\SU(2)$ satisfying $\|\widetilde A-A\|\le \epsilon^{2/15}/100$, there exists $B\in\SU(2)$ such that, for every $\widetilde B\in\SU(2)$ satisfying $\|\widetilde B-B\|\le \epsilon^{1/5}/100$, the unitary $\Pthirtytwo(\widetilde B\widetilde A,\widetilde A\widetilde B)$ is $(e,\epsilon)$-directed. The classical descriptions of $A$ and $B$ can be computed in $O(1)$ time, given $e,\epsilon$ and, for $B$, also $\widetilde A$.
\end{proposition}
The next lemma presents a cancellation property of the sequence $\Pthirtytwo(Y_1,Y_2)$ when $Y_1$ and $Y_2$ have opposite real parts.

\begin{lemma}\label{lem:directed_cancellation}
There exists an absolute constant $C>0$ such that, for every
$q_1,q_2\in\mathrm{Sp}(1)$ satisfying $\operatorname{Re}q_1=\tau=-\operatorname{Re}q_2$, setting $\delta=\operatorname{Re}(q_1q_2)+\tau^2$, we have
\[
\left\|
\Pthirtytwo(q_1,q_2)-(1,0)
-48\tau\delta^3(q_1+q_2)
\right\|
\le
C\max\{|\tau|^{17/3},\,|\delta|^{17/4}\}.
\]
\end{lemma}
\begin{proof}
We begin by showing that $\Pthirtytwo(q_1,q_2)$ can be expressed as a linear combination of $(1,0),q_1,q_2,q_1q_2$. Quaternion multiplication gives
\begin{equation}\label{eq:P32_quaternion_reductions}
\begin{cases}
q_1^2=2\tau q_1-(1,0),\\
q_2^2=-2\tau q_2-(1,0),\\
q_2q_1=2(\delta+\tau^2)(1,0)-2\tau q_1+2\tau q_2-q_1q_2.
\end{cases}
\end{equation}
Applying \eqref{eq:P32_quaternion_reductions} to the first block $q_1^3q_2$, we obtain
\[
\begin{aligned}
q_1^3q_2
=\bigl(2\tau q_1-(1,0)\bigr)q_1q_2
=2\tau\bigl(2\tau q_1-(1,0)\bigr)q_2-q_1q_2
=(4\tau^2-1)q_1q_2-2\tau q_2.
\end{aligned}
\]
After multiplying by each new factor, we repeatedly use \eqref{eq:P32_quaternion_reductions} to replace occurrences of $q_1^2$, $q_2^2$, and $q_2q_1$. The resulting expression contains only real multiples of the four quaternions $(1,0), q_1,q_2$, and $q_1q_2$. We denote the coefficients of these four terms by $\beta_0,\beta_1,\beta_2$, and $\beta_{12}$, respectively. Direct calculation establishes that
\[
\begin{aligned}
\beta_0=&
1+384\tau^2\delta^4-2560\tau^4\delta^3
+5888\tau^6\delta^2-5632\tau^8\delta+1920\tau^{10}
+O\big(|\tau|^{12}+|\delta|^6\big),\\
\beta_1,\beta_2=&
48\tau\delta^3-272\tau^3\delta^2
+464\tau^5\delta-240\tau^7
+O\big(|\tau|^9+|\delta|^{9/2}\big),\\
\beta_{12}=&
-96\tau^2\delta^3+544\tau^4\delta^2
-928\tau^6\delta+480\tau^8
+O\big(|\tau|^{10}+|\delta|^5\big),
\end{aligned}
\]
and therefore
\[
\begin{aligned}
\left\|
\Pthirtytwo(q_1,q_2)-(1,0)
-48\tau\delta^3(q_1+q_2)
\right\|&\le
|\beta_0-1|
+|\beta_1-48\tau\delta^3|
+|\beta_2-48\tau\delta^3|
+|\beta_{12}|\\
&\le
C\max\left\{|\tau|^{17/3},|\delta|^{17/4}\right\}
\end{aligned}
\]
for some absolute constant $C$.
\end{proof}

The following lemma establishes a geometric property of rotations that we will use to choose $B$ in the proof of Proposition~\ref{prop:directed_stick}.

\begin{lemma}\label{lem:rotation_geometry}
Let $R$ be a rotation of $\mathbb R^3$ of angle $\theta\in[\pi/2,\pi]$ about a unit axis $u$. Let $e\perp u$ be a unit vector. For every $\ell>0$ and $0\le\rho<\ell^2$, there exists $v\in\mathbb R^3$ satisfying
\[
\|v\|=\ell,
\qquad
v\cdot Rv=\rho,
\qquad
v-Rv=\sqrt{2(\ell^2-\rho)}\,e.
\]
Moreover, for every $\eta\ge0$ and every $\tilde v\in\mathbb R^3$,
\[
|\tilde v\cdot R\tilde v-\rho|
\le
2\|\tilde v-v\|\sqrt{-\ell^2\cos\theta+(1+\cos\theta)\rho}
+\|\tilde v-v\|^2.
\]
\end{lemma}

\begin{figure}[htbp!]
\centering
\includegraphics[width=0.6\linewidth]{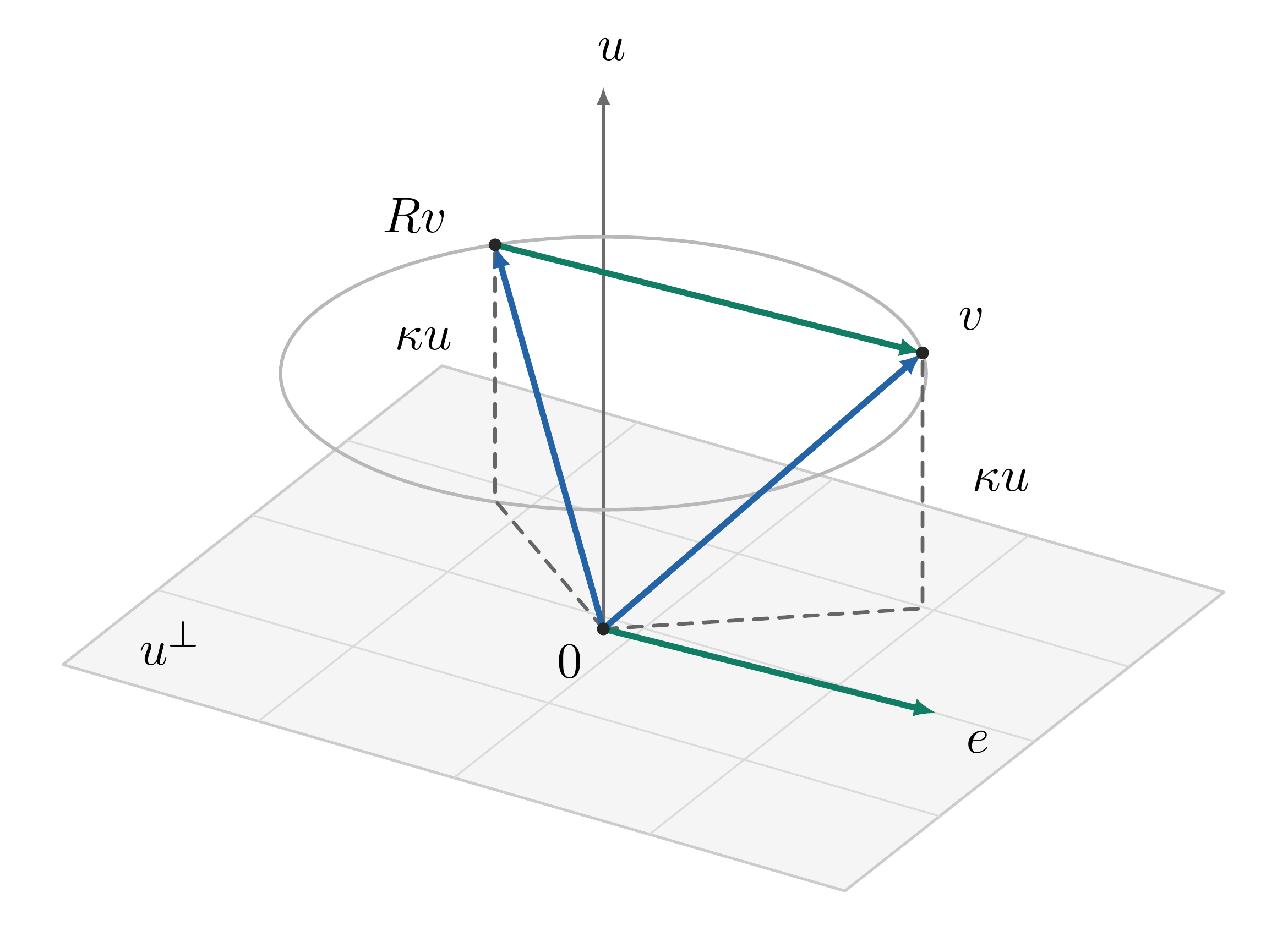}
\caption{The geometric intuition of Lemma~\ref{lem:rotation_geometry}}
\label{fig:rotation-geometry}
\end{figure}

\begin{proof}
Write $v=v_\perp+\kappa u$, where $v_\perp\in u^\perp$ and $\kappa\ge0$. Then, we have $v-Rv=v_\perp-Rv_\perp$ and
\[
\|v_\perp-Rv_\perp\|^2
=2\|v_\perp\|^2-2v_\perp\cdot Rv_\perp
=2(1-\cos\theta)\|v_\perp\|^2.
\]
Since $1-\cos\theta>0$, the linear map $v_\perp\mapsto v_\perp-Rv_\perp$ is invertible on $u^\perp$. Therefore, there exists $v_\perp$ such that
\[
v_\perp-Rv_\perp=\sqrt{2(\ell^2-\rho)}\,e.
\]
which satisfies
\[
\|v_\perp\|^2
=\frac{\ell^2-\rho}{1-\cos\theta}
\le\ell^2.
\]
Setting
\[
\kappa=\sqrt{\ell^2-\|v_\perp\|^2}
=\sqrt{\frac{\rho-\ell^2\cos\theta}{1-\cos\theta}}
\]
then gives
\[
\|v\|^2=\|v_\perp\|^2+\kappa^2=\ell^2,
\qquad
v\cdot Rv=\cos\theta\,\|v_\perp\|^2+\kappa^2=\rho.
\]
For the perturbation bound, note that we have
\[
Rv_\perp+R^{-1}v_\perp=2\cos\theta\,v_\perp.
\]
Consequently,
\[
\begin{aligned}
\tilde v\cdot R\tilde v-\rho
&=v\cdot R(\tilde v-v)+(\tilde v-v)\cdot Rv+(\tilde v-v)\cdot R(\tilde v-v)\\
&=2(\cos\theta\,v_\perp+\kappa u)\cdot(\tilde v-v)+(\tilde v-v)\cdot R(\tilde v-v).
\end{aligned}
\]
By our choices of $v_\perp$ and $\kappa$,
\[
\begin{aligned}
\|\cos\theta\,v_\perp+\kappa u\|^2
=\cos^2\theta\,\|v_\perp\|^2+\kappa^2
=-\ell^2\cos\theta+(1+\cos\theta)\rho.
\end{aligned}
\]
Then by the Cauchy-Schwartz inequality, we can conclude that
\[
\begin{aligned}
|\tilde v\cdot R\tilde v-\rho|
&\le 2\|\tilde v-v\|\,\|\cos\theta\,v_\perp+\kappa u\|+\|\tilde v-v\|^2\\
&=
2\|\tilde v-v\|\sqrt{-\ell^2\cos\theta+(1+\cos\theta)\rho}
+\|\tilde v-v\|^2.
\end{aligned}
\]
\end{proof}

\begin{proof}[Proof of Proposition~\ref{prop:directed_stick}]
We choose $A$ so that conjugation by every $\widetilde A$ that is close to $A$ induces a rotation of angle slightly larger than $\pi/2$. For each such $\widetilde A$, we choose a unit direction $\tilde e$ close to $e$ and perpendicular to the rotation axis, then use Lemma~\ref{lem:rotation_geometry} to choose an operator $B$ so that the leading term of $\log \Pthirtytwo(B\widetilde A,\widetilde A B)$ has the required size and points in direction $\tilde e$.

Let $u$ be the normalized projection onto $e^\perp$ of a coordinate unit vector whose inner product with $e$ has absolute value at most $3/4$. Set
\[
a\coloneqq
\left(
\sqrt{\frac{1-\epsilon^{2/15}}2},
\sqrt{\frac{1+\epsilon^{2/15}}2}\,u
\right),
\qquad
A\coloneqq\Phi(a).
\]
Write $\tilde a\coloneqq\Phi^{-1}(\widetilde A)$, so $\|\tilde a-a\|\le\epsilon^{2/15}/100$. By Lemma~\ref{fact:quaternion-rotation}, there exists a rotation $R$ induced by conjugation by $\tilde a$, whose angle $\theta\in[0,\pi]$ satisfies $\cos\theta=2(\operatorname{Re}\tilde a)^2-1$. Since $\|\operatorname{Im}a\|\ge1/\sqrt2$ and $2(\operatorname{Re}a)^2-1=-\epsilon^{2/15}$, the bound $\|\tilde a-a\|\le\epsilon^{2/15}/100$ gives, for sufficiently small $\epsilon_0$,
\[
\begin{aligned}
\|\tilde u-u\|\leq \frac{\epsilon^{2/15}}{25},
\qquad
-\frac{26}{25}\epsilon^{2/15}
\le\cos\theta
\le-\frac{24}{25}\epsilon^{2/15},
\end{aligned}
\]
where $\tilde u\coloneqq
\operatorname{Im}\tilde a/\|\operatorname{Im}\tilde a\|$ is the unit axis of $R$. The second estimate also gives $\theta\in(\pi/2,\pi)$. Since $e\perp u$, we have $|e\cdot\widetilde u|=O(\epsilon^{2/15})<1/2$. The normalized projection of $e$ onto $\widetilde u^\perp$ is therefore well-defined and satisfies
\[
\widetilde e\coloneqq
\frac{e-(e\cdot\widetilde u)\widetilde u}
{\|e-(e\cdot\widetilde u)\widetilde u\|},
\qquad
\|\widetilde e-e\|=O(\epsilon^{2/15}).
\]
Apply Lemma~\ref{lem:rotation_geometry} with $\ell=\sqrt{1-\epsilon^{2/5}}$ and $\rho=\epsilon^{1/5}\sqrt{-\cos\theta}$, which satisfy $0<\rho<\ell^2$ for sufficiently small $\epsilon_0$. For the resulting vector $v$, set $B=\Phi((\epsilon^{1/5},v)\tilde a^{-1})$.
For any $\widetilde B$ satisfying the stated error bound, write
\[
q_1=\Phi^{-1}(\widetilde B\widetilde A)=(\tau,\tilde v),
\qquad
q_2=-\Phi^{-1}(\widetilde A\widetilde B)=(-\tau,-R\tilde v).
\]
Then $|\tau-\epsilon^{1/5}|$ and $\|\tilde v-v\|$ are at most $\epsilon^{1/5}/100$. Put $\delta=\operatorname{Re}(q_1q_2)+\tau^2=\tilde v\cdot R\tilde v$. The perturbation bound in Lemma~\ref{lem:rotation_geometry} gives
\[
|\delta-\rho|
\le\frac{\epsilon^{1/5}}{50}\sqrt{-\ell^2\cos\theta+(1+\cos\theta)\rho}
+\frac{\epsilon^{2/5}}{10^4}
\le\frac{\rho}{40}
\]
for sufficiently small $\epsilon_0$. In particular, $99\epsilon^{1/5}/100\le\tau\le101\epsilon^{1/5}/100$ and $19\epsilon^{4/15}/20\le\delta\le21\epsilon^{4/15}/20$.
Also, $\tilde v-R\tilde v=\sqrt{2(\ell^2-\rho)}\,\tilde e+O(\epsilon^{1/5})$. There are sixteen copies of $q_2$ in $\Pthirtytwo(q_1,q_2)$, so their minus signs cancel. Lemma~\ref{lem:directed_cancellation} and Eq.~\eqref{eq:quaternion_principal_log} therefore give
\[
\log\Pthirtytwo(\widetilde B\widetilde A,\widetilde A\widetilde B)
=48\tau\delta^3\sqrt{2(\ell^2-\rho)}\,e+O(\epsilon^{17/15}).
\]
For sufficiently small $\epsilon_0$, the coefficient of $e$ lies between $50\epsilon$ and $80\epsilon$. The output is therefore $(e,\epsilon)$-directed.

Finally, we discuss the time complexity of this process. The construction in Lemma~\ref{lem:rotation_geometry} gives
\[
v=\sqrt{2(\ell^2-\rho)}\left((I-R)|_{\tilde u^\perp}\right)^{-1}\tilde e
+\sqrt{\frac{\rho-\ell^2\cos\theta}{1-\cos\theta}}\,\tilde u.
\]
Thus, the descriptions of $A$ and $B$ can be computed in $O(1)$ time using a constant number of arithmetic operations and roots of fixed degree, together with a linear solve in two dimensions.
\end{proof}

\subsection{A local correction net}\label{sec:local_correction_net}

Here we construct a local correction net from an $(e,\epsilon)$-directed unitary $S$ (Definition~\ref{def:directed_unitary}) and two unitaries $Q_1,Q_2$, called \textit{anchors}. For $C_1,C_2\in\mathrm{SU}(2)$, define the local correction sequence 
\begin{align}\label{eq:local_correction_sequence}
\mathcal P_{24}(S,C_1,C_2)
=(SC_1SC_1)(SC_2SC_1)(SC_2SC_2)(C_1SC_2C_1)(C_2SC_1C_1)(C_2C_1C_2C_2).
\end{align}
For suitable choices of $e$ and $Q_1,Q_2$, varying $C_1,C_2$ in neighborhoods of $Q_1,Q_2$ lets us approximate every unitary in a neighborhood of $\mathcal P_{24}(S,Q_1,Q_2)$. 
For simplicity, here we choose
\[
Q_1=iX,\qquad Q_2=iY,\qquad e=e_0\coloneqq (1,1,2)/\sqrt6.
\]
and $C_j\coloneqq V(h)Q_jV(h)^{-1}$ for $j=1,2$, where $V(x)\coloneqq\Phi(\exp(0,x))$ for $x\in\mathbb R^3$, and $h$ will be specified below. Write $s\coloneqq\log S$. For $S\ne I$ sufficiently close to $I$, define
\begin{equation}\label{eq:P24_conjugation_map}
\Psi_S(x)\coloneqq\|s\|^{-2}\bigl[\log\mathcal P_{24}(S,V(x)Q_1V(x)^{-1},V(x)Q_2V(x)^{-1})
-\log\mathcal P_{24}(S,Q_1,Q_2)\bigr].
\end{equation}
The sequence equals $U$ at $x=h$ precisely when $\Psi_S(h)=(\log U-\log\mathcal P_{24}(S,Q_1,Q_2))/\|s\|^2$. We choose $h$ as the unique fixed point in $B_0(1/1000)$ of the map
\begin{equation}\label{eq:P24_fixed_point_update}
x\longmapsto x+(D\Psi_S(0))^{-1}\left(\frac{\log U-\log\mathcal P_{24}(S,Q_1,Q_2)}{\|s\|^2}-\Psi_S(x)\right).
\end{equation}
The following proposition establishes that this choice is well defined and bounds the error caused by approximating $C_1$ and $C_2$.

\begin{proposition}\label{prop:local_correction_net}
There exist constants $c,\epsilon_0>0$ such that, for every $0<\epsilon<\epsilon_0$, every $(e_0,\epsilon)$-directed unitary $S$, and every $U\in\mathrm{SU}(2)$ satisfying $\|U-\mathcal P_{24}(S,Q_1,Q_2)\|\le c\epsilon^2$, the fixed point $h$ of the map in \eqref{eq:P24_fixed_point_update} is well defined with $\|h\|\le1$, can be computed in time $\poly\log(1/\epsilon)$, and satisfies
\[
\mathcal P_{24}(S,C_1,C_2)=U.
\]
for $C_j\coloneqq V(h)Q_jV(h)^{-1}$, $j=1,2$. Moreover, for every $\widetilde C_1,\widetilde C_2\in\mathrm{SU}(2)$ satisfying $\|\widetilde C_j-C_j\|\le c\epsilon$ for $j=1,2$, we have $\|\mathcal P_{24}(S,\widetilde C_1,\widetilde C_2)-U\|\le\epsilon^3$.
\end{proposition}

\subsubsection{Expansion and error bounds}

We compute the quadratic term of $\log\mathcal P_{24}(S,Q_1,Q_2)$ and bound the remainder and its first three derivatives. We then bound how $\mathcal P_{24}$ changes when $Q_1,Q_2$ are replaced by nearby $C_1,C_2$.

\begin{lemma}\label{lem:local_correction_cancellation}
There exists an absolute constant $\epsilon_0>0$ such that, for every $0<\epsilon<\epsilon_0$ and every $S\in\mathrm{SU}(2)\setminus\{-I\}$ satisfying $\|\log S\|\le100\epsilon$, the sequence $\mathcal P_{24}$ in \eqref{eq:local_correction_sequence} satisfies
\begin{align}\label{eq:P24-expansion}
\log\mathcal P_{24}(S,Q_1,Q_2)
=4s_1(0,s_3,s_2)+O(\|s\|^3).
\end{align}
\end{lemma}

\begin{proof}
Since $Q_1^2=Q_2^2=-I$ and $Q_1Q_2=-Q_2Q_1$,
\[
\mathcal P_{24}(I,Q_1,Q_2)
=(Q_2Q_1^2Q_2)(Q_1Q_2^2Q_1)=I.
\]
Applying $Q_jU=(Q_jUQ_j^{-1})Q_j$ for $U\in\mathrm{SU}(2)$ in \eqref{eq:local_correction_sequence} and using $\mathcal P_{24}(I,Q_1,Q_2)=I$ gives
\[
\begin{aligned}
\mathcal P_{24}(S,Q_1,Q_2)
&=\bigl[S(Q_1SQ_1^{-1})\bigr]\bigl[S(Q_2SQ_2^{-1})\bigr]\cdot\bigl[(Q_1Q_2S(Q_1Q_2)^{-1})(Q_1SQ_1^{-1})\bigr]\\
&\qquad\quad\cdot\bigl[(Q_2SQ_2^{-1})(Q_1Q_2S(Q_1Q_2)^{-1})\bigr].
\end{aligned}
\]
In the quaternion notation, we have
\[
\begin{aligned}
&\Phi^{-1}\!\bigl(S(Q_1SQ_1^{-1})\bigr)
=\bigl(1-2s_1^2,\,(2s_1,\,2s_1s_3,\,-2s_1s_2)\bigr)+O(\|s\|^3),\\
&\Phi^{-1}\!\bigl(S(Q_2SQ_2^{-1})\bigr)
=\bigl(1-2s_2^2,\,(-2s_2s_3,\,2s_2,\,2s_1s_2)\bigr)+O(\|s\|^3),\\
&\Phi^{-1}\!\bigl((Q_1Q_2S(Q_1Q_2)^{-1})(Q_1SQ_1^{-1})\bigr)
=\bigl(1-2s_2^2,\,(2s_2s_3,\,-2s_2,\,2s_1s_2)\bigr)+O(\|s\|^3),\\
&\Phi^{-1}\!\bigl((Q_2SQ_2^{-1})(Q_1Q_2S(Q_1Q_2)^{-1})\bigr)
=\bigl(1-2s_1^2,\,(-2s_1,\,2s_1s_3,\,2s_1s_2)\bigr)+O(\|s\|^3).
\end{aligned}
\]
which leads to
\[
\Phi^{-1}\!\left(\mathcal P_{24}(S,Q_1,Q_2)\right)
=\bigl(1,4s_1(0,s_3,s_2)\bigr)+O(\|s\|^3),
\]
and
\[
\log\mathcal P_{24}(S,Q_1,Q_2)
=4s_1(0,s_3,s_2)+O(\|s\|^3).
\]
\end{proof}

The next lemma analyzes the higher order error terms in Lemma~\ref{lem:local_correction_cancellation} and its derivatives. For $y\in B_0(1/16)$, define
\begin{equation}\label{eq:P24_remainder}
E(y)\coloneqq\log\mathcal P_{24}(V(y),Q_1,Q_2)-4y_1(0,y_3,y_2).
\end{equation}

\begin{lemma}\label{lem:local_correction_remainder}
For every $y\in B_0(1/16)$,
\[
\|E(y)\|=O(\|y\|^3),\qquad
\|DE(y)\|=O(\|y\|^2),\qquad
\|D^2E(y)\|=O(\|y\|),\qquad
\|D^3E(y)\|=O(1).
\]
\end{lemma}

\begin{proof}
We first bound $D^3E$. For $y\in B_0(1/16)$, the exponential series gives
\[
\|D^k\exp(0,y)\|
\le\sum_{n=k}^{\infty}\frac{\|y\|^{n-k}}{(n-k)!}
=e^{\|y\|},\qquad 0\le k\le3.
\]
Since $S$ occurs eight times in \eqref{eq:local_correction_sequence}, the product rule gives
\begin{equation}\label{eq:P24_remainder_product_derivatives}
\bigl\|D_y^k[\Phi^{-1}(\mathcal P_{24}(V(y),Q_1,Q_2))]\bigr\|
\le8^k e^{8\|y\|}\le2\cdot8^k,\qquad k=1,2,3.
\end{equation}
Lemma~\ref{lem:local_correction_cancellation} with $S=I$ gives $\mathcal P_{24}(I,Q_1,Q_2)=I$, and
\[
\|\Phi^{-1}(\mathcal P_{24}(V(y),Q_1,Q_2))-(1,0)\|
\le8\|\exp(0,y)-(1,0)\|\le8\|y\|\le\frac12.
\]
The logarithm series gives
\[
E(y)=\operatorname{Im}\sum_{n=1}^{\infty}\frac{(-1)^{n+1}}{n}
\bigl[\Phi^{-1}(\mathcal P_{24}(V(y),Q_1,Q_2))-(1,0)\bigr]^n-4y_1(0,y_3,y_2).
\]
Differentiating this series three times and applying \eqref{eq:P24_remainder_product_derivatives} gives
\[
\begin{aligned}
\|D^3E(y)\|
&\le(2\cdot8^3)\sum_{n=1}^{\infty}2^{1-n}
+3(2\cdot8)(2\cdot8^2)\sum_{n=2}^{\infty}(n-1)2^{2-n}\\
&\qquad+(2\cdot8)^3\sum_{n=3}^{\infty}(n-1)(n-2)2^{3-n}
<10^5.
\end{aligned}
\]
Lemma~\ref{lem:local_correction_cancellation} gives $E(y)=O(\|y\|^3)$ as $y\to0$, hence $E(0)=DE(0)=D^2E(0)=0$. Taylor's formula gives the stated bounds on $E$, $DE$, and $D^2E$ throughout $B_0(1/16)$.
\end{proof}

The next lemma bounds the effect of replacing the anchors $Q_1,Q_2$ by nearby unitaries in the value of $\mathcal P_{24}$.

\begin{lemma}\label{lem:local_correction_error_control}
There exists an absolute constant $\epsilon_0>0$ such that, for every $0<\epsilon<\epsilon_0$, every $S\in\mathrm{SU}(2)\setminus\{-I\}$ satisfying $\|\log S\|\le100\epsilon$, every $c\in(0,1]$, and every $C_1,C_2\in\mathrm{SU}(2)$ satisfying $\|C_1-Q_1\|,\|C_2-Q_2\|\le c\epsilon$, we have
\begin{align}\label{eq:P24_error_control}
\|\mathcal P_{24}(S,C_1,C_2)-\mathcal P_{24}(S,Q_1,Q_2)\|
=O(c\epsilon^3).
\end{align}
The implied constant is absolute.
\end{lemma}

\begin{proof}
Write $s\coloneqq\log S$ and $\alpha_j\coloneqq\log(Q_j^{-1}C_j)$ for $j\in\{1,2\}$. Then $C_j=Q_j\Phi(\exp(0,\alpha_j))$, $\|\alpha_j\|\le2\|C_j-Q_j\|\le2c\epsilon$ for $\epsilon<\epsilon_0\le1$, and
\begin{equation}\label{eq:P24_input_series}
\Phi^{-1}(S)=\sum_{k=0}^{\infty}\frac{(0,s)^k}{k!},
\qquad
\Phi^{-1}(C_j)=\Phi^{-1}(Q_j)\sum_{k=0}^{\infty}\frac{(0,\alpha_j)^k}{k!}.
\end{equation}
Substituting these series into \eqref{eq:local_correction_sequence} and collecting terms of degree at most two gives
\begin{equation}\label{eq:P24_joint_expansion}
\Phi^{-1}\!\left(\mathcal P_{24}(S,C_1,C_2)\right)
=\bigl(1,4s_1(0,s_3,s_2)\bigr)
+O\!\left((\|s\|+\|\alpha_1\|+\|\alpha_2\|)^3\right).
\end{equation}
Subtracting $\Phi^{-1}(\mathcal P_{24}(S,Q_1,Q_2))$ removes all terms independent of $\alpha_1,\alpha_2$, including everything of degree at most two. To bound the difference term by term, apply $\Phi^{-1}$ to \eqref{eq:local_correction_sequence} and expand each factor using \eqref{eq:P24_input_series}. Let $T_{a,b}$ be the product of one chosen summand from each series, in the original order, where the lists $a=(a_1,\ldots,a_8)$ and $b=(b_1,\ldots,b_{16})$ record the powers chosen at the eight occurrences of $\Phi^{-1}(S)$ and the sixteen occurrences of $\Phi^{-1}(C_1),\Phi^{-1}(C_2)$, respectively, from left to right.
For example, choosing the linear summand for the first $\Phi^{-1}(S)$, the quadratic summand for the last $\Phi^{-1}(C_2)$, and the constant summands elsewhere gives
\[
T_{(1,0,\ldots,0),(0,\ldots,0,2)}
=(0,s)\frac{(0,\alpha_2)^2}{2!},
\]
since Lemma~\ref{lem:local_correction_cancellation} with $S=I$ gives $\mathcal P_{24}(I,Q_1,Q_2)=I$.

The norm of $T_{a,b}$ is the product of the norms of its factors. Using $\|s\|\le100\epsilon$, $\|\alpha_j\|\le2c\epsilon$, and $\|\Phi^{-1}(Q_j)\|=1$, we obtain
\[
\|T_{a,b}\|
\le\frac{(100\epsilon)^{\sum_i a_i}(2c\epsilon)^{\sum_i b_i}}
{\prod_{i=1}^{8}a_i!\,\prod_{i=1}^{16}b_i!}.
\]
The remaining terms have $\sum_i b_i\ge1$ and $\sum_i a_i+\sum_i b_i\ge3$. Applying the triangle inequality after the cancelations then gives
\begin{equation}\label{eq:P24_series_term_bound}
\begin{aligned}
\|\mathcal P_{24}(S,C_1,C_2)-\mathcal P_{24}(S,Q_1,Q_2)\|
&=\Biggl\|\sum_{\substack{\sum_i b_i\ge1\\\sum_i a_i+\sum_i b_i\ge3}}T_{a,b}\Biggr\|
\le\sum_{\substack{\sum_i b_i\ge1\\\sum_i a_i+\sum_i b_i\ge3}}\|T_{a,b}\|\\
&\le\sum_{\substack{\sum_i b_i\ge1\\\sum_i a_i+\sum_i b_i\ge3}}
\frac{(100\epsilon)^{\sum_i a_i}(2c\epsilon)^{\sum_i b_i}}
{\prod_{i=1}^{8}a_i!\,\prod_{i=1}^{16}b_i!}\\
&=\sum_{k=3}^{\infty}\frac{((800+32c)\epsilon)^k-(800\epsilon)^k}{k!}
=O(c\epsilon^3).
\end{aligned}
\end{equation}
\end{proof}

\subsubsection{Local surjectivity}

We now prove local surjectivity of $\Psi_S$ and existence of the fixed point $h$ defined above. We first bound the derivatives of the approximation obtained from the quadratic term in Lemma~\ref{lem:local_correction_cancellation}. Let $R_x$ be the rotation induced by $\exp(0,x)$ from Lemma~\ref{fact:quaternion-rotation}, and define
\[
\overline{\Psi}_S(x)\coloneqq\frac{4}{\|s\|^2}\Bigl[
R_x\bigl((R_x^{-1}s)_1(0,(R_x^{-1}s)_3,(R_x^{-1}s)_2)\bigr)-s_1(0,s_3,s_2)\Bigr].
\]

\begin{lemma}\label{lem:local_correction_derivatives}
For every $\epsilon>0$ and every $(e_0,\epsilon)$-directed unitary $S$, the derivative $D\overline{\Psi}_S(0)$ is invertible, with $\|(D\overline{\Psi}_S(0))^{-1}\|\le3$ and $\|D^2\overline{\Psi}_S(x)\|\le72$ for every $x\in B_0(1)$.
\end{lemma}

\begin{proof}
By Lemma~\ref{fact:quaternion-rotation}, $R_xv=v+2x\times v+O(\|x\|^2\|v\|)$ for every $v\in\mathbb R^3$. Therefore,
\begin{equation}\label{eq:P24_conjugation_derivative}
D\overline{\Psi}_S(0)
=\frac{8}{\|s\|^2}
\begin{pmatrix}
0&s_1s_2&-s_1s_3\\
-2s_1s_2&s_1^2-s_3^2&s_2s_3\\
2s_1s_3&-s_2s_3&s_2^2-s_1^2
\end{pmatrix}.
\end{equation}
Since $S$ is $(e_0,\epsilon)$-directed, $\|s/\|s\|-e_0\|\le1/20$, which implies $1/3<s_1/\|s\|<1/2$, $(s_3-s_2)/\|s\|>3/10$, and $(s_3+s_2)/\|s\|>1$. Consequently,
\[
|\det D\overline{\Psi}_S(0)|
=\frac{1024s_1^2(s_3^2-s_2^2)(\|s\|^2-2s_1^2)}{\|s\|^6}
>\frac{256}{15}.
\]
Let $\sigma_1\ge\sigma_2\ge\sigma_3>0$ be the singular values of $D\overline{\Psi}_S(0)$. Summing the squares of the entries in \eqref{eq:P24_conjugation_derivative} gives $\sigma_1^2+\sigma_2^2+\sigma_3^2=64(1+s_1^2/\|s\|^2)<80$. Hence,
\[
\|(D\overline{\Psi}_S(0))^{-1}\|
=\frac{\sigma_1\sigma_2}{|\det D\overline{\Psi}_S(0)|}
\le\frac{\sigma_1^2+\sigma_2^2}{2|\det D\overline{\Psi}_S(0)|}
<3.
\]
To bound the second derivative, we may assume $\|s\|=1$, since $\overline{\Psi}_S$ depends only on $s/\|s\|$. Set $F(y)=y_1(0,y_3,y_2)$. For every unit vector $y$, direct calculation gives $\|F(y)\|\le1/2$, $\|DF(y)\|\le1$, and $\|D^2F(y)\|\le1$. For $y,v\in\mathbb R^3$, differentiating the rotation equation $\partial_tR_{tx}y=2x\times R_{tx}y$ gives
\[
DR_x[v]y=2\int_0^1 R_{(1-t)x}\bigl(v\times R_{tx}y\bigr)\,\d t.
\]
Rotations preserve norms, so
\begin{equation}\label{eq:rotation_derivative_bounds}
\|DR_x\|\le2,\qquad \|D^2R_x\|\le4.
\end{equation}
Since $R_x^{-1}=R_{-x}$, Eq.~\eqref{eq:rotation_derivative_bounds} also holds with $R_x$ replaced by $R_x^{-1}$. Applying the product and chain rules to $\overline{\Psi}_S(x)=4[R_xF(R_x^{-1}s)-F(s)]$ yields
\[
\|D^2\overline{\Psi}_S(x)\|
\le4\left(\tfrac12\|D^2R_x\|+2\|DR_x\|\,\|DR_x^{-1}\|+\|DR_x^{-1}\|^2+\|D^2R_x^{-1}\|\right)
\le72.
\]
\end{proof}

\begin{lemma}\label{lem:local_correction_coverage}
There exists an absolute constant $\epsilon_0>0$ such that, for every $0<\epsilon<\epsilon_0$ and every $(e_0,\epsilon)$-directed unitary $S$, the derivative $D\Psi_S(0)$ is invertible, with $\|(D\Psi_S(0))^{-1}\|\le5$ and $\|D^2\Psi_S(x)\|\le100$ for every $x\in B_0(1)$.
\end{lemma}

\begin{proof}
Take $\epsilon_0\le1/1600$. By the definition of $\mathcal P_{24}$ in Eq.~\eqref{eq:local_correction_sequence},
\begin{equation}\label{eq:P24_common_conjugation}
\mathcal P_{24}(S,V(x)Q_1V(x)^{-1},V(x)Q_2V(x)^{-1})
=V(x)\mathcal P_{24}(V(x)^{-1}SV(x),Q_1,Q_2)V(x)^{-1}.
\end{equation}
By Lemma~\ref{fact:quaternion-rotation},
\[
\log(V(x)^{-1}SV(x))=R_x^{-1}s,
\qquad \|R_x^{-1}s\|=\|s\|\le100\epsilon<1/16.
\]
Therefore, by the definition of $E$ in Eq.~\eqref{eq:P24_remainder}, we have
\[
\log\mathcal P_{24}(V(x)^{-1}SV(x),Q_1,Q_2)
=4(R_x^{-1}s)_1(0,(R_x^{-1}s)_3,(R_x^{-1}s)_2)+E(R_x^{-1}s).
\]
Together with \eqref{eq:P24_common_conjugation}, this leads to
\begin{equation}\label{eq:P24_conjugation_remainder}
\Psi_S(x)-\overline{\Psi}_S(x)
=\frac{R_xE(R_x^{-1}s)-E(s)}{\|s\|^2}.
\end{equation}
Differentiating the conjugation formula in Lemma~\ref{fact:quaternion-rotation} gives, for $x\in B_0(1)$,
\[
\|DR_x\|+\|D^2R_x\|=O(1),
\qquad
\|D_x(R_x^{-1}s)\|+\|D_x^2(R_x^{-1}s)\|=O(\|s\|).
\]
Lemma~\ref{lem:local_correction_remainder} and the chain rule give $\|E(R_x^{-1}s)\|+\|D_x[E(R_x^{-1}s)]\|=O(\|s\|^3)$ and
\[
\bigl\|D_x^2[E(R_x^{-1}s)]\bigr\|
\le\|D^2E(R_x^{-1}s)\|\,\|D_x(R_x^{-1}s)\|^2
+\|DE(R_x^{-1}s)\|\,\|D_x^2(R_x^{-1}s)\|
=O(\|s\|^3).
\]
Differentiating \eqref{eq:P24_conjugation_remainder} now gives
\begin{equation}\label{eq:P24_derivative_comparison}
\|D\Psi_S(0)-D\overline{\Psi}_S(0)\|=O(\|s\|),
\qquad
\sup_{x\in B_0(1)}\|D^2\Psi_S(x)-D^2\overline{\Psi}_S(x)\|=O(\|s\|),
\end{equation}
with absolute implied constants.
We choose $\epsilon_0$ so that the two bounds in Eq.~\eqref{eq:P24_derivative_comparison} are at most $1/10$ and $1$, respectively, which is possible because $\|s\|\le100\epsilon<100\epsilon_0$. Lemma~\ref{lem:local_correction_derivatives} then gives
\[
\|D\Psi_S(0)v\|\ge\|D\overline{\Psi}_S(0)v\|-\frac{\|v\|}{10}\ge\frac{\|v\|}{5},
\qquad v\in\mathbb R^3.
\]
Consequently, $\|(D\Psi_S(0))^{-1}\|\le5$ and $\|D^2\Psi_S(x)\|\le100$ for $\|x\|\le1$.
\end{proof}

\begin{proof}[Proof of Proposition~\ref{prop:local_correction_net}]
We first establish the existence and uniqueness of $h$. Choose $0<c\le1/40000$ and $\epsilon_0>0$ small enough for Lemmas~\ref{lem:local_correction_cancellation}--\ref{lem:local_correction_coverage} to apply. Set $u\coloneqq(\log U-\log\mathcal P_{24}(S,Q_1,Q_2))/\|s\|^2$, which satisfies
\[
\|u\|\le\frac{2\|U-\mathcal P_{24}(S,Q_1,Q_2)\|}{\|s\|^2}\le2c.
\]
For every $x\in B_0(1/1000)$, the derivative bounds in Lemma~\ref{lem:local_correction_coverage} give
\[
\|I-(D\Psi_S(0))^{-1}D\Psi_S(x)\|
\le5\|D\Psi_S(x)-D\Psi_S(0)\|
\le500\|x\|\le\frac12.
\]
Since $\Psi_S(0)=0$, we further have
\[
\left\|x+(D\Psi_S(0))^{-1}(u-\Psi_S(x))\right\|
\le5\|u\|+\frac12\|x\|\le\frac{3}{4000}<\frac{1}{1000}.
\]
Moreover, $\forall x,y\in B_0(1/1000)$, we have
\begin{equation}\label{eq:P24_contraction_estimate}
\begin{aligned}
&\left\|x-y-(D\Psi_S(0))^{-1}\bigl(\Psi_S(x)-\Psi_S(y)\bigr)\right\|\\
&\qquad=\left\|\int_0^1\left[I-(D\Psi_S(0))^{-1}D\Psi_S((1-t)y+tx)\right](x-y)\,\d t\right\|\\
&\qquad\le\int_0^1\frac12\|x-y\|\,\d t=\frac12\|x-y\|.
\end{aligned}
\end{equation}
Therefore, applying Lemma~\ref{lem:banach_fixed_point} with $D=B_0(1/1000)$ and $q=1/2$ gives a unique fixed point $h$ in this ball. Its fixed-point equation is $\Psi_S(h)=u$, which gives $\mathcal P_{24}(S,C_1,C_2)=U$.

For the perturbation bound, suppose $\|\widetilde C_j-C_j\|\le c\epsilon$ for $j=1,2$. Then, we have $\|V(h)^{-1}\widetilde C_jV(h)-Q_j\|=\|\widetilde C_j-C_j\|\le c\epsilon$ and $\|\log(V(h)^{-1}SV(h))\|=\|s\|\le100\epsilon$. Lemma~\ref{lem:local_correction_error_control} applied to these matrices gives, after conjugation by $V(h)$,
\begin{equation}\label{eq:P24_anchor_approximation_error}
\|\mathcal P_{24}(S,\widetilde C_1,\widetilde C_2)-U\|
=O(c\epsilon^3).
\end{equation}
Choosing a small enough $c$ makes the bound in Eq.~\eqref{eq:P24_anchor_approximation_error} at most $\epsilon^3$.

In the rest of the proof, we discuss the time complexity for computing $h$. To approximate $h$, we iterate the map in Eq.~\eqref{eq:P24_fixed_point_update}:
\[
h_0=0,\qquad
h_{k+1}=h_k+(D\Psi_S(0))^{-1}\bigl(u-\Psi_S(h_k)\bigr),
\qquad k\ge0.
\]
Since $\Psi_S(h)=u$, Eq.~\eqref{eq:P24_contraction_estimate} gives
\[
\|h_{k+1}-h\|
=\left\|h_k-h-(D\Psi_S(0))^{-1}\bigl(\Psi_S(h_k)-\Psi_S(h)\bigr)\right\|
\le\frac12\|h_k-h\|.
\]
Consequently, we have $\|h_k-h\|\le2^{-k}\|h\|$ and
\[
\|V(h_k)Q_jV(h_k)^{-1}-C_j\|
\le2\|V(h_k)-V(h)\|
\le2\|h_k-h\|
\le2^{1-k},
\qquad j=1,2.
\]
Taking $k=\lceil\log_2(2/(c\epsilon))\rceil=O(\log(1/\epsilon))$ gives $\|V(h_k)Q_jV(h_k)^{-1}-C_j\|\le c\epsilon$ for $j=1,2$. Applying Eq.~\eqref{eq:P24_anchor_approximation_error} with $\widetilde C_j=V(h_k)Q_jV(h_k)^{-1}$ therefore gives
\[
\left\|\mathcal P_{24}\bigl(S,V(h_k)Q_1V(h_k)^{-1},V(h_k)Q_2V(h_k)^{-1}\bigr)-U\right\|
\le\epsilon^3.
\]
The initial computation of $u$ and $(D\Psi_S(0))^{-1}$ and each subsequent update take time $\poly\log(1/\epsilon)$, giving the claimed runtime.
\end{proof}

\subsection{The compilation algorithms}\label{sec:full_compilation}

We now show how to use gates from $G$ to approximate a target unitary to the desired accuracy. We start with a coarse approximation and improve it by appending a correction constructed from a directed unitary. Algorithm~\ref{alg:full_compile} uses Proposition~\ref{prop:local_correction_net} to construct the correction, while Algorithm~\ref{alg:directed_unitary} uses Proposition~\ref{prop:directed_stick} to supply the directed unitary. Both algorithms obtain the coarser approximations they need recursively. Throughout, we fix constants $0<c\le1$ and $0<\epsilon_0<\min\{c/2,1/1000\}$ small enough for Propositions~\ref{prop:directed_stick} and~\ref{prop:local_correction_net} to apply.

\begin{algorithm}[htbp!]
\caption{\textsc{FullCompile}$(U,\epsilon)$}\label{alg:full_compile}
\begin{algorithmic}[1]
\Require target $U\in\SU(2)$; accuracy $0<\epsilon<1$
\Ensure a $G$-word $w$ with $\|W-U\|\le\epsilon$
\lIf{$\epsilon\ge\epsilon_0^3$}{\Return $\epsilon_0^3\text{-}\textsc{ConstantNet}(U)$}
\State $w_S\gets\Call{DirectedUnitary}{e_0,\epsilon^{1/3}}$
\State $S\gets\val(w_S)$
\State $w_0\gets\Call{FullCompile}{U\mathcal P_{24}(S,Q_1,Q_2)^{-1},c\epsilon^{2/3}}$
\State Compute $h$ using Proposition~\ref{prop:local_correction_net} with inputs $S,W_0^{-1}U,\epsilon^{1/3}$.
\State $V\gets\Phi(\exp(0,h))$
\For{$j\in\{1,2\}$}
  \State $w_j\gets\Call{FullCompile}{VQ_jV^{-1},c\epsilon^{1/3}}$
\EndFor
\State \Return $w_0\mathcal P_{24}(w_S,w_1,w_2)$
\end{algorithmic}
\end{algorithm}

\begin{algorithm}[htbp!]
\caption{\textsc{DirectedUnitary}$(e,\epsilon)$}\label{alg:directed_unitary}
\begin{algorithmic}[1]
\Require unit vector $e\in\mathbb R^3$; scale $0<\epsilon<\epsilon_0$
\Ensure a $G$-word $w$ such that $W$ is $(e,\epsilon)$-directed
\State Compute $A$ from $e,\epsilon$ using Proposition~\ref{prop:directed_stick}.
\State $w_A\gets\Call{FullCompile}{A,\epsilon^{2/15}/100}$
\State Compute $B$ from $e,\epsilon,\widetilde A=W_A$ using Proposition~\ref{prop:directed_stick}.
\State $w_B\gets\Call{FullCompile}{B,\epsilon^{1/5}/100}$
\State \Return $\Pthirtytwo(w_Bw_A,w_Aw_B)$
\end{algorithmic}
\end{algorithm}

Write $L_{\mathrm{full}}(\epsilon),T_{\mathrm{full}}(\epsilon)$ for the worst-case output length and runtime of Algorithm~\ref{alg:full_compile} for $0<\epsilon<1$, and $L_{\mathrm{dir}}(\epsilon),T_{\mathrm{dir}}(\epsilon)$ for those of Algorithm~\ref{alg:directed_unitary} for $0<\epsilon<\epsilon_0$. For the base case $\epsilon_0^3\le\epsilon<1$, Lemma~\ref{lem:compilation_base_case} gives $L_{\mathrm{full}}(\epsilon)=O(1)$ and $T_{\mathrm{full}}(\epsilon)=O(1)$. The next two lemmas give recursive bounds on the output lengths and runtime of both algorithms.

\begin{lemma}\label{lem:directed_unitary_recursion}
Let $e\in\mathbb R^3$ be a unit vector and $0<\epsilon<\epsilon_0$. Algorithm~\ref{alg:directed_unitary} returns a $G$-word $w$ such that $W$ is $(e,\epsilon)$-directed. Its length and runtime satisfy
\[
\begin{aligned}
L_{\mathrm{dir}}(\epsilon)
&\le32L_{\mathrm{full}}(\epsilon^{2/15}/100)
+32L_{\mathrm{full}}(\epsilon^{1/5}/100),\\
T_{\mathrm{dir}}(\epsilon)
&\le T_{\mathrm{full}}(\epsilon^{2/15}/100)
+T_{\mathrm{full}}(\epsilon^{1/5}/100)
+L_{\mathrm{dir}}(\epsilon)\poly\log(1/\epsilon).
\end{aligned}
\]
\end{lemma}

\begin{proof}
The correctness claim of Algorithm~\ref{alg:directed_unitary} follows from Proposition~\ref{prop:directed_stick}. Note that the output $w=\Pthirtytwo(w_Bw_A,w_Aw_B)$ contains 32 copies of each of $w_A$ and $w_B$, so
\[
\len(w)=32\len(w_A)+32\len(w_B)
\le32L_{\mathrm{full}}(\epsilon^{2/15}/100)+32L_{\mathrm{full}}(\epsilon^{1/5}/100).
\]
For the runtime, the calls producing $w_A$ and $w_B$ cost at most $T_{\mathrm{full}}(\epsilon^{2/15}/100)+T_{\mathrm{full}}(\epsilon^{1/5}/100)$ in total. Computing $A$ and $B$ takes $O(1)$ time by Proposition~\ref{prop:directed_stick}, while evaluating $W_A$ and computing $w$ takes $O(\len(w)\poly\log(1/\epsilon))$ time. Since $\len(w)\le L_{\mathrm{dir}}(\epsilon)$, adding these costs gives the claimed bound.
\end{proof}

\begin{lemma}\label{lem:full_compile_recursion}
Let $U\in\mathrm{SU}(2)$ and $0<\epsilon<\epsilon_0^3$. Algorithm~\ref{alg:full_compile} returns a $G$-word $w$ satisfying $\|W-U\|\le\epsilon$. Its length and runtime satisfy
\[
\begin{aligned}
L_{\mathrm{full}}(\epsilon)
&\le L_{\mathrm{full}}(c\epsilon^{2/3})
+8L_{\mathrm{dir}}(\epsilon^{1/3})
+16L_{\mathrm{full}}(c\epsilon^{1/3}),\\
T_{\mathrm{full}}(\epsilon)
&\le T_{\mathrm{full}}(c\epsilon^{2/3})
+T_{\mathrm{dir}}(\epsilon^{1/3})
+2T_{\mathrm{full}}(c\epsilon^{1/3})
+L_{\mathrm{full}}(\epsilon)\poly\log(1/\epsilon).
\end{aligned}
\]
\end{lemma}

\begin{proof}
Lemma~\ref{lem:directed_unitary_recursion} gives an $(e_0,\epsilon^{1/3})$-directed unitary $S=\val(w_S)$. Moreover,
\[
\|W_0^{-1}U-\mathcal P_{24}(S,Q_1,Q_2)\|
=\|U\mathcal P_{24}(S,Q_1,Q_2)^{-1}-W_0\|
\le c\epsilon^{2/3}.
\]
Proposition~\ref{prop:local_correction_net} therefore applies with target $W_0^{-1}U$ and scale $\epsilon^{1/3}<\epsilon_0$, so the vector $h$ used in Algorithm~\ref{alg:full_compile} is well defined. For $V=\Phi(\exp(0,h))$, the recursive calls producing $w_1,w_2$ give $\|W_j-VQ_jV^{-1}\|\le c\epsilon^{1/3}$ for $j=1,2$, which leads to
\[
\|W-U\|
=\|\mathcal P_{24}(S,W_1,W_2)-W_0^{-1}U\|
\le\epsilon.
\]
For the sequence length, note that $\mathcal P_{24}$ contains each of its three arguments eight times. Hence the output $w=w_0\mathcal P_{24}(w_S,w_1,w_2)$ satisfies
\[
\begin{aligned}
\len(w)
&=\len(w_0)+8\len(w_S)+8\len(w_1)+8\len(w_2)\\
&\le L_{\mathrm{full}}(c\epsilon^{2/3})+8L_{\mathrm{dir}}(\epsilon^{1/3})+16L_{\mathrm{full}}(c\epsilon^{1/3}).
\end{aligned}
\]
For the runtime, the call producing $w_S$ and the three calls producing $w_0,w_1,w_2$ cost at most
\[
T_{\mathrm{dir}}(\epsilon^{1/3})+T_{\mathrm{full}}(c\epsilon^{2/3})+2T_{\mathrm{full}}(c\epsilon^{1/3})
\]
in total. Computing $\val(w_S)$ and $W_0$ and writing $w$ take $O(\len(w)\poly\log(1/\epsilon))$ time. Computing $h$ takes $\poly\log(1/\epsilon)$ time by Proposition~\ref{prop:local_correction_net}, as do the fixed number of matrix operations used to form $V$ and the target matrices. Since $\len(w)\le L_{\mathrm{full}}(\epsilon)$, adding these costs gives the claimed runtime bound.
\end{proof}

\begin{proof}[Proof of Theorem~\ref{thm:general_inverse_free}]
The guarantee that $\|W-U\|\le\epsilon$ follows directly from Lemmas~\ref{lem:compilation_base_case} and~\ref{lem:full_compile_recursion}. As for the sequence length, combining Lemmas~\ref{lem:directed_unitary_recursion} and~\ref{lem:full_compile_recursion} gives
\begin{equation}\label{eq:full_compile_combined_length}
\begin{aligned}
L_{\mathrm{full}}(\epsilon)
\le L_{\mathrm{full}}(c\epsilon^{2/3})
+16L_{\mathrm{full}}(c\epsilon^{1/3})
+256L_{\mathrm{full}}(\epsilon^{2/45}/100)
+256L_{\mathrm{full}}(\epsilon^{1/15}/100).
\end{aligned}
\end{equation}
Note that the exponent $2.988$ satisfies
\[
\left(\frac23\right)^{2.988}
+16\left(\frac13\right)^{2.988}
+256\left(\frac{2}{45}\right)^{2.988}
+256\left(\frac1{15}\right)^{2.988}
<0.999886<1.
\]
Therefore, for all sufficiently small $\epsilon$,
\[
\begin{aligned}
&\log^{2.988}\left(\frac{1}{c\epsilon^{2/3}}\right)
+16\log^{2.988}\left(\frac{1}{c\epsilon^{1/3}}\right)\\
&\qquad+256\log^{2.988}\left(\frac{100}{\epsilon^{2/45}}\right)
+256\log^{2.988}\left(\frac{100}{\epsilon^{1/15}}\right)
<\log^{2.988}\left(\frac1\epsilon\right).
\end{aligned}
\]
Thus, the induction in Eq.~\eqref{eq:full_compile_combined_length} gives
\begin{equation}\label{eq:full_compile_length_bound}
L_{\mathrm{full}}(\epsilon)\le O(\log^{2.988}(1/\epsilon)),\qquad 0<\epsilon\le e^{-1}.
\end{equation}

For the runtime, a call to Algorithm~\ref{alg:full_compile} with $\epsilon<\epsilon_0^3$ makes three recursive calls directly and two more through Algorithm~\ref{alg:directed_unitary}. Using Eq.~\eqref{eq:full_compile_length_bound}, Lemmas~\ref{lem:directed_unitary_recursion} and~\ref{lem:full_compile_recursion} give
\[
\begin{aligned}
T_{\mathrm{full}}(\epsilon)
\le T_{\mathrm{full}}(c\epsilon^{2/3})
+2T_{\mathrm{full}}(c\epsilon^{1/3})
+T_{\mathrm{full}}(\epsilon^{2/45}/100)
+T_{\mathrm{full}}(\epsilon^{1/15}/100)
+\poly\log(1/\epsilon).
\end{aligned}
\]
With $c$ fixed, decrease $\epsilon_0$ if necessary so that, for every $0<\epsilon<\epsilon_0^3$,
\[
\min\left\{c\epsilon^{2/3},c\epsilon^{1/3},\frac{\epsilon^{2/45}}{100},\frac{\epsilon^{1/15}}{100}\right\}
\ge\epsilon^{3/4}.
\]
Starting from input error $\epsilon$, every call at depth $k$ therefore has error tolerance at least $\epsilon^{(3/4)^k}$. All branches have reached the base case once $\epsilon^{(3/4)^k}\ge\epsilon_0^3$, which takes $O(\log\log(1/\epsilon))$ rounds. Since each node has at most five children, the recursion tree has at most $5^{O(\log\log(1/\epsilon))}=\poly\log(1/\epsilon)$ nodes. Multiplying by the $\poly\log(1/\epsilon)$ nonrecursive work per node gives $T_{\mathrm{full}}(\epsilon)=\poly\log(1/\epsilon)$.
\end{proof}

\bibliographystyle{alpha}
\bibliography{main}

\end{document}